\documentclass[preprint,12pt]{elsarticle}

\usepackage{subfigure}
\usepackage{caption}
\usepackage{amssymb}
\usepackage{amsfonts}
\usepackage{amsmath}
\usepackage{amsthm}
\usepackage{amssymb}
\usepackage{multirow}
\usepackage{booktabs}
\usepackage{array}
\usepackage{colortbl}

\DeclareMathOperator*{\argmin}{argmin}
\newtheorem{theorem}{Theorem}

\newdefinition{remark}{Remark}

\usepackage{algorithm}
\usepackage{algpseudocode}
\journal{Elsevier}

\begin{document}

\begin{frontmatter}



%\title{A Modified Method for Recommendation System Based on Sampling Reconstruction of Signal on Graphs}
\title{An Equivalent Graph Reconstruction Model and its Application in Recommendation Prediction}

\author[label1]{Guangrui Yang}
\ead{gyang29@cityu.edu.hk}
 \author[label2,label3]{Lihua Yang}
 \ead{mcsylh@mail.sysu.edu.cn}
 \author[label2]{Qing Zhang}
 \ead{ zhangq369@mail2.sysu.edu.cn}
 \author[label4]{Zhihua Yang\corref{cor1}}
 \cortext[cor1]{Corresponding author}
 \ead{yangzh@ghufe.edu.cn}
 \address[label1]{Department of Mathematics, City University of Hong Kong, Hong Kong}
 \address[label2]{School of Mathematics, Sun Yat-sen University, Guangzhou, 510275, China}
 \address[label3]{Guangdong Province Key Laboratory of Computational Science}
\address[label4]{School of Information Science, Guangdong University of Finance and Economics, Guangzhou, 510320, China}
%%%\tnotetext[]{This research was partially supported by National Natural Science Foundation of China (Nos: 12171488), Guangdong Province Key Laboratory of Computational Science at the Sun %%%Yat-sen University (2020B1212060032).}

\begin{abstract}
Recommendation algorithm plays an important role in recommendation system (RS), which predict users' interests and preferences for some given items based on their known information. Recently, a recommendation algorithm based on graph Laplacian regularization was proposed, which treats the prediction problem of the recommendation system as a reconstruction issue of small samples of the graph signal under the same graph model. Such a technique takes into account both known and unknown labeled samples information, thereby obtaining good prediction accuracy. However, when the data size is large, solving the reconstruction model is computationally expensive even with an approximate strategy. In this paper, we propose an equivalent reconstruction model that can be solved exactly with extremely low computational cost. Finally, a final prediction algorithm is proposed. We find in the experiments that the proposed method significantly reduces the computational cost while maintaining a good prediction accuracy.
\end{abstract}




\end{frontmatter}


\section{Introduction}
%With the rapid development of mobile internet, recommendation system (RS) has a widely application in e-commerce. Using the recommendation system algorithm, the network platform can automatically recommend the content that may be of interest to the users according to their preferences. The e-commerce platforms, such as Amazon, Youtube, Netflix and Yahoo, all benefit from RS. Today, in the context of big data, the analysis of recommendation system algorithms has become a more active research field.
Recommendation algorithm plays an important role in recommendation system (RS), which predict users' interests and preferences for some given items based on their known information. %Recently, with the rapid development of e-commerce, the research on recommendation algorithms has received widespread attention.
As a consequence, recommendation algorithms are widely used by major electronic platforms (e.g. Amazon, Youtube, and Netflix) to recommend products of interest to users, and have received widespread attention.

%Rating prediction is one of recommendation algorithms in RS. It focuses on predicting users' ratings for items based on few known ratings.
The core of the recommendation algorithm is rating prediction, that is, predicting the ratings of all users for all items.
Specifically, if we denote the user set and the item set by $\mathcal{U}$ and $\mathcal{I}$, then the issue of rating prediction becomes to predict a score matrix $\mathbf{S}=\Big(s(u,i)\Big)_{u\in\mathcal{U},i\in\mathcal{I}}\in\mathbb{R}^{|\mathcal{U}|\times|\mathcal{I}|}$, where $s(u,i)$ is the rating given by user $u$ for item $i$ and $|\cdot|$ indicates the cardinality of a set. In real-world applications, often only few users' ratings can be observed, i.e., only few ratings in $\mathbf{S}$ are known. Therefore, we need to predict the unknown ratings of $\mathbf{S}$ according to the known ratings, so as to obtain a complete score matrix $\mathbf{S}$ \cite{2011RS}. Once all ratings in $\mathbf{S}$ are predicted, recommendations can be implemented, i.e., items with high ratings that may be of interest to users can be recommended.
%The researches on RS can be traced to the middle of 1990s \cite{1992CF,1998RS_classification}.
%The core of the prediction techniques is to evaluate a user's interest in an item based on some existing information \cite{2011RS}. Mathematically, the goal of RS is to predict a score matrix $\mathbf{S}$. Such an $\mathbf{S}$ is the collection of all $s(u, i)$ where $s(u, i)$ is the rating given by user $u$ for item $i$, i.e., $\mathbf{S}=\Big(s(u,i)\Big)_{u\in\mathcal{U},i\in\mathcal{I}}\in\mathbb{R}^{|\mathcal{U}|\times|\mathcal{I}|}$, where $\mathcal{U}$ and $\mathcal{I}$ are the sets of the users and items respectively and $|\cdot|$ indicates the cardinality of a set. In the real-world applications, the matrix $\mathbf{S}$ is usually sparse (only few entries are known), since a user often does not rate all items and an item is often not rated by all users. Therefore, the unknown entries in $\mathbf{S}$ need to be predicted according to the known entries. Once all entries of $\mathbf{S}$ are predicted, recommendations can be implemented, i.e., items (with high ratings) that may be of interest to users can be recommended.

The rating prediction problem of $\mathbf{S}$ has received attention since the 1990s \cite{1992CF,1998RS_classification}. Over the past few decades, a variety of different prediction methods have emerged. A common traditional prediction method is matrix completion (MC) \cite{2008Cai,2010MC_WithNoise,2010Zhang_MC}.  Assuming that $\mathbf{S}$ is a low rank matrix, MC estimates the unknown ratings of $\mathbf{S}$ by solving a low-rank recovery problem. Compared to MC, collaborative filtering (CF) is another traditional prediction technique \cite{1994CF_Architecture,1999CF_algorithm,2011CF,2014CF_Bigdata}. It first establishes the similarity between users (or items), and then predicts the unknown rating of user $u$ for item $i$ based on the ratings of other users who are most similar to user $u$. That is, the key idea of CF is to infer the preference of an active user towards a given item based on the opinions of some similar-minded users in the system. % However, the traditional matrix completion and collaborative filtering methods usually discard the unknown entries in $\mathbf{S}$ which actually contain a lot of useful information for prediction.
Furthermore, by treating the ratings denoted by integers as labels, the prediction problem of $\mathbf{S}$ can also be regarded as a classification problem. Therefore, besides MC and CF, some classical classification methods in traditional machine learning are encouraged to predict $\mathbf{S}$ in RSs \cite{2002WebMining,2011feedbackRecommendation,2011Recommender,2002CollaborativeBayesianClassifier,2006SVM-basedTVpro}. %For example, \citep{2002WebMining} used the decision tree as a filter to select which users can be the recommendation target. Literatures \citep{2011feedbackRecommendation,2011Recommender} applied logistic regression to their recommendation systems. \cite{2002CollaborativeBayesianClassifier} implemented a recommendation system based on a naive Bayesian classifier. And \cite{2006SVM-basedTVpro} used SVM to recommend TV programs.

However, MC, CF and traditional machine learning methods discard a large amount of unlabeled data, which often contain a lot of useful information for prediction \cite{2006Manifold,2013Hady}, and thus perform poorly in terms of prediction accuracy. %In recent years, great progress has been made in semi-supervised learning. Researchers in \cite{2006Manifold} and \cite{2013Hady} found that learning accuracy can be significantly improved when unlabeled data are also involved in the learning process.
To improve prediction accuracy, the authors in \cite{2020Yang} considered using the information from both labeled and unlabeled data, and thus proposed a prediction model based on graph Laplacian regularization in a Reproducing Kernel Hilbert Space (RKHS) to solve the prediction problem (see model \eqref{Ori_model}).

More specifically, \cite{2020Yang} constructed a graph model to mine the inherent information of the data. Under this graph model, users (or items) and their relationships are modeled as nodes and edge weights of graph, respectively. Then, based on the assumption that similar users (items) have similar ratings, each column (or row) of $\mathbf{S}$ can be regarded as a smooth signal on the graph.%Then, the prediction problem of sparse $\mathbf{S}$ can be treated as a reconstruction issue of small samples of the graph signal. %the feature vectors for each vertex in the feature space.
%More specifically, \cite{2020Yang} modeled the users (user-based) or the items (item-based) and their relationship as a weighted undirected graph to effectively capture the potential structure of the data. With this model, the prediction problem of sparse $\mathbf{S}$ can be treated as a reconstruction issue of small samples of the graph signal. On this basis, the authors in \cite{2020Yang} first constructed the feature vectors for each vertex in the feature space. Then, according to the manifold assumption \cite{2013Semi-SupervisedLearning}, the mappings of all the vertices under a mapping $\Phi$ are usually distributed on a low dimensional manifold in the high dimensional feature space that is easy to classify. Therefore, \cite{2020Yang} proposed a prediction model based on graph Laplacian regularization in a Reproducing Kernel Hilbert Space (RKHS) to solve the prediction problem.

In the view of graph signal processing, such a prediction technique in \cite{2020Yang} treats the prediction problem of $\mathbf{S}$ as a reconstruction issue of small samples (i.e., known ratings) of the graph signal under the same graph model. Finally, when considering the user-based prediction (user as node), the prediction problem for each item is equivalent to a quadratic unconditional optimization problem. Then, it requires $O(n^{3})$ computations to solve analytical the quadratic unconditional optimization problem, where $n=|\mathcal{U}|$ is the number of users. For $m$ items, the optimization problem should be solved $m$ times and thus, a total of $O(n^{3}m)$ computations are required to predict the complete score matrix $\mathbf{S}$, which is computationally expensive for large $n$ and $m$. To reduce the computational complexity, \cite{2020Yang} proposed an approximate solution strategy to solve the original reconstruction model based on the reconstruction of bandlimited graph signal \cite{2013Shuman_gsp,2016Vertex,2014DGSP}, which reduces the computational cost from $O(n^{3}m)$ to $O(n^{2}m(k_{b}+\ell))$ to predict the last $m-1$ items$\footnote{Note that both methods require $O(n^{3})$ computations to predict the first item.}$, where $k_b$ is the reconstruction bandwidth and $\ell$ is the number of known labels$\footnote{When considering that $k_{b}$ and $\ell$ are much smaller than $n$, $O(n^{2}m(k_{b}+\ell))$ can be written as $O(n^{2}m)$, the same as that in \cite{2020Yang}. But no such considerations are made in this paper.}$.

However, this approximate method just obtains an approximate solution to the original reconstruction model. Moreover, when the data size is very large, i.e., both $n$ and $m$ are very large, it is still computationally expensive. Based on the idea in \cite{2020Yang}, we realized that in order to address the multi-item prediction problem in RS more efficiently, it is necessary to avoid solving higher-order linear equations. For this purpose, to exactly solve the original reconstruction model and meanwhile reduce the computational complexity, we modify the original reconstruction model of \cite{2020Yang} to obtain its equivalent reconstruction model. In this sense, the key contributions of this paper are summarized as follows:

\begin{enumerate}
  \item  An equivalent reconstruction model with that of \cite{2020Yang} is proposed, which allows us to find solutions of the reconstruction model in a much lower-dimensional subspace.
  \item A strategy is designed skillfully to exactly solve the reconstruction model with extremely low computational cost.
  \item A new recommendation algorithm based on the proposed equivalent reconstruction model is designed, which leads excellent experimental results both on prediction accuracy and computational complexity.
%  \item A final prediction algorithm based on the proposed equivalent prediction model is completed for recommendation system. The experimental results show that the proposed method maintains a good prediction accuracy.
\end{enumerate}

The rest of this paper is organized as follow. In Section $\ref{Sec_primevalmodel}$, we first introduce the original prediction model based on graph Laplacian regularization in recommendation system. In Section $\ref{Sec_modifiedmodel}$, an equivalent prediction model of the original model is proposed. Then, an efficient method is designed skillfully to solve the original model accurately and further reduces the computational cost. At the end of this section, we provide a computational complexity analysis of the proposed method. In Section $\ref{Sec_fianl_prediction_alg}$, a final prediction algorithm based on the proposed equivalent prediction model is completed for recommendation system. In Section $\ref{Sec_exp}$, several experiments are conducted to test the proposed technology. Finally, a brief conclusion of this paper is given in Section $\ref{Sec_conclusion}$.

\section{Preliminary}\label{Sec_primevalmodel}
\subsection{Existing model}
In this section, we first introduce the original prediction model proposed in \cite{2020Yang}. For the sake of simplicity, we focus the used-based prediction, and the results can be generalized to the item-based prediction.

%Given a RS with $n$ users, based on the assumption in \cite{2020Yang}, all users can be mapped into a feature space, i.e.,
%$$\Phi:\mathcal{U}\rightarrow\mathbb{R}^{d},$$
 %where $\mathcal{U}=\{1,2,...,n\}$ is the set of all users and $d$ represents the dimension of the feature space.
 Given a RS with $n$ users, and suppose that each user corresponds to a $d$-dimensional feature vector $\mathbf{v}_{j}$ uniquely \cite{2020Yang}. We denote by $\mathcal{V}=\{\mathbf{v}_{1},\mathbf{v}_{2},...,\mathbf{v}_{n}\}$ the all users represented in the feature space. For an item, the ratings given by all users can be viewed as a function on the set $\mathcal{V}$, i.e.,
 $$f:~\mathcal{V}\rightarrow\mathbb{R},$$
 where $f(\mathbf{v}_{j})$ denotes the ratings given by the user $\mathbf{v}_{j}$ for this item, which is also called the label of user $\mathbf{v}_{j}$. In this sense, $$\mathbf{f}=(f(\mathbf{v}_{1}),f(\mathbf{v}_{2}),\dots,f(\mathbf{v}_{n}))^{T}\in\mathbb{R}^{n}$$ can be viewed as a graph signal defined over $\mathcal{V}$. Then, the prediction problem becomes:   given labeled samples $X_{\ell}=\{(\mathbf{v}_{1},y_{1}),(\mathbf{v}_{2},y_{2}), . . . , (\mathbf{v}_{\ell},y_{\ell})\}$, how to predict the ratings of other users$\footnote{Without loss of generality, we assume that the ratings given by the first $\ell$ users are known.}$, i.e., how to find

 $$f:\mathcal{V}\rightarrow\mathbb{R},~~\textrm{s.t.}~~f(\mathbf{v}_{j})=y_{j},~j=1,2,\dots,\ell.$$

  To address this issue, \cite{2020Yang} proposed the following prediction model based on the manifold assumption and the smoothness assumption on the graph:

 \begin{align}\label{Ori_model}
 \argmin_{f\in\mathcal{H}_{\mathcal{K}}}\sum_{j=1}^{\ell}(f(\mathbf{x}_{j})-y_{j})^{2}+\lambda\|f\|^{2}_{\mathcal{K}}+\gamma\mathbf{f}^{T}\mathbf{L}\mathbf{f},
\end{align}
 where $\mathcal{H}_{\mathcal{K}}$ is the Reproducing Kernel Hilbert Space (RKHS) with respect to the kernel function $\mathcal{K}: \mathcal{V}\times\mathcal{V}\rightarrow\mathbb{R}$, i.e.,
 $$\mathcal{H}_{\mathcal{K}}:=\Big\{\sum_{j=1}^{n}a_{j}\mathcal{K}(\cdot,\mathbf{v}_{j}),~~a_{1},a_{2},\dots,a_{n}\in\mathbb{R}\Big\},$$
 with
 $$\big\langle\sum_{j=1}^{n}a_{j}\mathcal{K}(\cdot,\mathbf{v}_{j}),\sum_{j=1}^{n}b_{j}\mathcal{K}(\cdot,\mathbf{v}_{j})\big\rangle_{\mathcal{K}}
 =\sum_{i=1}^{n}\sum_{j=1}^{n}a_{i}b_{j}\mathcal{K}(\mathbf{v}_{i},\mathbf{v}_{j}).$$
  $\|f\|_{\mathcal{K}}^{2}:=\langle f,f\rangle_{\mathcal{K}},~\forall f\in\mathcal{H}_{\mathcal{K}}$. $\mathbf{f}=(f(\mathbf{v}_{1}),f(\mathbf{v}_{2}),\dots,f(\mathbf{v}_{n}))^{T}\in\mathbb{R}^{n}$. While $\mathbf{L}$ is a graph Laplacian matrix which reflects the intrinsic geometrical structure of $\mathcal{V}$ in the feature space. Specifically, the user set $\mathcal{V}$ can be modeled as a simple undirected graph $\mathcal{G}$, where nodes represent users and edges represent relationships between nodes. Each edge of $\mathcal{G}$ is assigned a weight to reflect the strength of the association between the connected nodes. Let $\mathbf{W}$ denote the weight matrix of all edges of $\mathcal{G}$, then $$\mathbf{L}:=\mathbf{D}-\mathbf{W},$$
  where $\mathbf{D}$ is the diagonal degree matrix with its diagonal elements $d_{i}=\sum_{j=1}^{n}W_{ij}$.

  It can be seen that the objective function of model \eqref{Ori_model} consists of three terms, namely: $\sum_{j=1}^{\ell}(f(\mathbf{v}_{j})-y_{j})^{2}$, $\lambda\|f\|_{\mathcal{K}}^{2}$ and $\gamma\mathbf{f}^{T}\mathbf{Lf}$, where the parameters $\lambda,\gamma>0$. Obviously, the first term $\sum_{j=1}^{\ell}(f(\mathbf{v}_{j})-y_{j})^{2}$ ensures that the function values of $f$ on the nodes with known ratings are close to the real values. The second term $\lambda\|f\|_{\mathcal{K}}^{2}$ characterizes the continuity of $f$ in the feature mapping space $\Phi(\mathcal{V})$, where $\Phi$ is a feature mapping that maps $\mathcal{V}$ into an implicit and easy-to-classify manifold \cite{2020Yang}, since
  \begin{align*}
    |f(\mathbf{v}_{i})-f(\mathbf{v}_{j})| &\leq\|f\|_{\mathcal{K}}\cdot\|\mathcal{K}(\cdot,\mathbf{v}_{i})-\mathcal{K}(\cdot,\mathbf{v}_{j})\|_{\mathcal{K}}  \\
     & =\|f\|_{\mathcal{K}}\cdot\|\Phi(\mathbf{v}_{i})-\Phi(\mathbf{v}_{j})\|_{\mathcal{K}}.
  \end{align*}
  While the last term $\gamma\mathbf{f}^{T}\mathbf{Lf}$ ensures that $f$ is smooth on graph $\mathcal{G}$, since
$$
\mathbf{f}^{T}\mathbf{Lf}=\frac{1}{2}\sum_{i,j=1}^{n}W_{ij}\Big(f(\mathbf{x}_{i})-f(\mathbf{x}_{j})\Big)^{2},
$$
where $W_{ij}$ is the $(i,j)$ element of $\mathbf{W}$ that reflects the affinity between users $i$ and $j$.

\subsection{Model solution}
The main goal is now to solve the prediction model $\eqref{Ori_model}$. First, according to \cite{2020Yang}, the minimizer of model $\eqref{Ori_model}$ admits an expansion
$$
 f^{*}(\mathbf{v})=\sum_{j=1}^{n}a_{j}\mathcal{K}(\mathbf{v}_{j},\mathbf{v}),~~\forall\mathbf{v}\in\mathcal{V},
$$
where $a_{j}\in\mathbb{R},~\forall j=1,2,...,n$. Let $\mathbf{a}=(a_{1},a_{2},...,a_{n})^{T}\in\mathbb{R}^{n}$, then solving model \eqref{Ori_model} is equivalent to solving the following linear equation:
\begin{equation}
  (\mathbf{K}_{\ell}\mathbf{K}_{\ell}^{T}+\lambda\mathbf{K}+\gamma\mathbf{KLK})\mathbf{a}=\mathbf{K}_{\ell}\mathbf{y}_{\ell}, \label{a_solution}
\end{equation}
where $\mathbf{y}_{\ell}=(y_{1},y_{2},...,y_{\ell})^{T}\in\mathbb{R}^{\ell}$, $\mathbf{K}\in\mathbb{R}^{n\times n}$ is the kernel gram matrix with $K_{ij}= \mathcal{K}(\mathbf{x}_{i},\mathbf{x}_{j})$ and $\mathbf{K}_{\ell}$ represents the sub-matrix consisting of the first $\ell$ columns of $\mathbf{K}$.

By solving $\mathbf{a}=(a_{1},a_{2},...,a_{n})^{T}\in\mathbb{R}^{n}$ based on \eqref{a_solution},
the value of $f^{*}$ at $\mathcal{V}$ can be obtained as $$\mathbf{f}^{*}=(f^{*}(\mathbf{v}_{1}),f^{*}(\mathbf{v}_{2}),\dots,f^{*}(\mathbf{v}_{n}))^{T}=\mathbf{K}\mathbf{a}.$$

Note that for each item, the prediction ratings given by all users can be obtained by solving $\eqref{a_solution}$. To distinguish this method from the subsequent approximation method and the proposed method, we call it  the original method (labeled $``\mathbf{Ori}"$). It is clear that $\mathbf{Ori}$ needs $O(n^{3})$ computations to achieve the prediction for one item by solving $\eqref{a_solution}$. Furthermore, since matrix $\mathbf{K}_{\ell}\mathbf{K}_{\ell}^{T}+\lambda\mathbf{K}+\gamma\mathbf{KLK}$ and vector $\mathbf{K}_{\ell}\mathbf{y}_{\ell}$ vary with the label samples,  a total of $O(n^{3}m)$ computations are required to predict all unknown labels of $m$ items by solving $\eqref{a_solution}$ $m$ times. Therefore, it is computationally expensive when $m$ and $n$ are both large.

To reduce the computational complexity, \cite{2020Yang} provided an approximate solution to the optimization $\eqref{a_solution}$ based on bandlimited reconstruction on graphs. Such a graph-based approximate method (labeled $``\mathbf{GBa}"$) treats the solution of \eqref{a_solution} approximately as a bandlimited signal on graph $\mathcal{G}$, thereby avoiding solving \eqref{a_solution} repeatedly. Specifically, $\mathbf{a}$ is written approximately as
$$\mathbf{a}\approx\mathbf{U}_{k_{b}}\mathbf{c},$$
where $k_{b}$ is the reconstructed bandwidth, $\mathbf{c}\in\mathbb{R}^{k_{b}}$ and $\mathbf{U}_{k_{b}}$ denotes the sub-matrix consisting of the first $k_{b}$ columns of the eigenvector matrix $\mathbf{U}\in\mathbb{R}^{n\times n}$ of $\mathbf{L}$. According to \cite{2020Yang}, $\mathbf{c}$ can be solved by the following linear equation:
\begin{equation}\label{c_solution}
   \Big(\mathbf{U}_{k_{b}}^{T}(\mathbf{K}_{\ell}\mathbf{K}_{\ell}^{T}+\lambda\mathbf{K}+\gamma\mathbf{KLK})\mathbf{U}_{k_{b}}\Big)\mathbf{c}
   =\mathbf{U}_{k_{b}}^{T}\mathbf{K}_{\ell}\mathbf{y}_{\ell}.
\end{equation}
Finally, $\mathbf{GBa}$ obtains an approximate solution of the original model \eqref{Ori_model} by solving \eqref{c_solution}, i.e., $\mathbf{f}^{*}\approx\mathbf{K}\mathbf{U}_{k_{b}}\mathbf{c}$.

Compared to $\mathbf{Ori}$, $\mathbf{GBa}$ only needs to solve a $k_{b}\times k_{b}$ linear equation to predict each item, thereby making it more efficient when $k_{b}\ll n$. Since the eigendecomposition of the graph Laplacian $\mathbf{L}$ only needs to be computed once, $\mathbf{GBa}$ reduces the computational cost from $O(n^{3}m)$ to $O(n^{2}m(k_{b}+\ell))$ to predict the last $m-1$ items.

However, $\mathbf{GBa}$ only obtains an approximate solution of the original model. Although this approximate solution will approach the optimal solution as $k_{b}$ increases, but at the same time it requires more computational cost. To accurately solve the original model and further reduce the computational cost, we propose an equivalent prediction model of the original model $\eqref{Ori_model}$ in the next section. Such an equivalent prediction model allows us to find a solution of model $\eqref{Ori_model}$ in a low-dimensional subspace and thus needs much less computational cost.

\section{An equivalent prediction model}\label{Sec_modifiedmodel}
%Our solution is to modify the original prediction model \eqref{Ori_model} to obtain its equivalent prediction model. For this purpose, let
To propose the equivalent prediction model, let
\begin{equation}\label{New_inner_product}
\langle f,g\rangle_{\mathcal{R}}:=\lambda\langle f,g\rangle_{\mathcal{K}}+\gamma\mathbf{f}^{T}\mathbf{L}\mathbf{g},~~\forall f,g\in\mathcal{H}_{\mathcal{K}},
\end{equation}
 where $\mathbf{f}=(f(\mathbf{v}_{1}),f(\mathbf{v}_{2}),...,f(\mathbf{v}_{n}))^{T},~\mathbf{g}=(g(\mathbf{v}_{1}),g(\mathbf{v}_{2}),...,g(\mathbf{v}_{n}))^{T}\in\mathbb{R}^{n}$. Then, it can be easily seen that $\langle\cdot,\cdot\rangle_{\mathcal{R}}$ redefines an inner product on $\mathcal{H}_{\mathcal{K}}$ and satisfies
 \begin{equation}\label{R_norm}
 \|f\|_{\mathcal{R}}^{2}:=\lambda\|f\|^{2}_{\mathcal{K}}+\gamma\mathbf{f}^{T}\mathbf{L}\mathbf{g}\geq\lambda\|f\|_{\mathcal{K}}^{2},~~\forall f\in\mathcal{H}_{\mathcal{K}}.
 \end{equation}
It thus means that after reequipping with inner product $\langle \cdot,\cdot\rangle_{\mathcal{R}}$, $\mathcal{H}_{\mathcal{K}}$ constitutes a new Hilbert space, and we denote it by $\mathcal{H}_{\mathcal{R}}$.

\subsection{The proposed model and its solution}
Obviously, $\mathcal{H}_{\mathcal{K}}$ and $\mathcal{H}_{\mathcal{R}}$ are two Hilbert spaces with the same vector space but equipped with different inner products. Therefore, the original prediction model \eqref{Ori_model} is equivalent to the following prediction model:
\begin{equation}\label{Equivalent_model}
\argmin_{f\in\mathcal{H}_{\mathcal{R}}}\sum_{j=1}^{\ell}(f(\mathbf{v}_{j})-y_{j})^{2}+\|f\|^{2}_{\mathcal{R}}.
\end{equation}
%where $\|\cdot\|_{\mathcal{R}}$ is the RKHS norm of $\mathcal{H}_{\mathcal{R}}$ with respect to the inner product $\langle \cdot,\cdot\rangle_{\mathcal{R}}$.
It shows that solving the original model \eqref{Ori_model} is equivalent to solving the model \eqref{Equivalent_model}.

To solve model \eqref{Equivalent_model}, we first prove the following theorem.
\begin{theorem}
  $\mathcal{H}_{\mathcal{R}}$ is a reproducing kernel Hilbert space (RKHS). \label{New_RKHS}
\end{theorem}
\begin{proof}
  According to the results in \cite{1950TheoryOfRK}, since $\mathcal{H}_{\mathcal{K}}$ is an RKHS, for any $\mathbf{v}\in\mathcal{V}$, there exists a constant $M_{\mathbf{v}}>0$ such that
  $$|f(\mathbf{v})|\leq M_{\mathbf{v}}\|f\|_{\mathcal{K}},~~\forall f\in\mathcal{H}_{\mathcal{K}}.$$
  Furthermore, by \eqref{R_norm}, we have
  $$|f(\mathbf{v})|\leq M_{\mathbf{v}}\|f\|_{\mathcal{K}}\leq M_{\mathbf{v}}\lambda^{-1/2}\|f\|_{\mathcal{R}},~~\forall f\in\mathcal{H}_{\mathcal{R}},$$
  where means that for any $\mathbf{v}\in\mathcal{V}$, the point functional $$\delta_{\mathbf{v}}:f\rightarrow f(\mathbf{v}),~~\forall f\in\mathcal{H}_{\mathcal{R}}$$ is a continuous linear functional on $\mathcal{H}_{\mathcal{R}}$. Therefore, $\mathcal{H}_{\mathcal{R}}$ is also an RKHS, which completes the proof.
\end{proof}

The above theorem shows that $\mathcal{H}_{\mathcal{R}}$ is an RKHS, and thus we can next apply the representer theorem of RKHS to solve the model \eqref{Equivalent_model}. Specifically, let us denote by $\mathcal{R}: \mathcal{V}\times\mathcal{V}\rightarrow\mathbb{R}$ the kernel function of $\mathcal{H}_{\mathcal{R}}$ and $\mathbf{R}\in\mathbb{R}^{n\times n}$ the kernel gram matrix with $R_{ij}= \mathcal{R}(\mathbf{v}_{i},\mathbf{v}_{j})$. Then, according to the representer theorem \cite{2006Manifold,2014ProgressiveImage}, the minimizer of model $\eqref{Equivalent_model}$ admits an expansion

 \begin{equation}\label{Solution_R}
 f^{*}(\mathbf{v})=\sum_{j=1}^{\ell}d_{j}\mathcal{R}(\mathbf{v}_{j},\mathbf{v}),~~\forall \mathbf{v}\in\mathcal{V},
 \end{equation}

where $d_{j}\in\mathbb{R}$ for $j=1,2,...,\ell$. Obviously, for any $f$ satisfying form \eqref{Solution_R}, we have

$$\sum_{j=1}^{\ell}(f(\mathbf{v}_{j})-y_{j})^{2}=\|\mathbf{y}_{\ell}-\mathbf{R}_{\ell,\ell}\mathbf{d}\|_{2}^{2},$$
and
\begin{align*}
  \|f\|^{2}_{\mathcal{R}} & =\langle f,f\rangle_{\mathcal{R}}
  =\Big\langle\sum_{j=1}^{\ell}d_{j}\mathcal{R}(\mathbf{v}_{j},\mathbf{v}),\sum_{j=1}^{\ell}d_{j}\mathcal{R}(\mathbf{v}_{j},\mathbf{v}) \Big\rangle_{\mathcal{R}} \\
                    &  =\sum_{i=1}^{\ell}\sum_{j=1}^{\ell}d_{i}d_{j}\mathcal{R}(\mathbf{v}_{i},\mathbf{v}_{j})=\mathbf{d}^{T}\mathbf{R}_{\ell,\ell}\mathbf{d},
\end{align*}
where $\mathbf{d}=(d_{1},d_{2},...,d_{\ell})^{T}\in\mathbb{R}^{\ell}$ and $\mathbf{R}_{\ell,\ell}$ is the sub-matrix consisting of the first $\ell$ rows and the first $\ell$ columns of $\mathbf{R}$.

Therefore, solving model \eqref{Equivalent_model} is equivalent to solving the following optimization problem:
\begin{equation}
\argmin_{\mathbf{d}\in\mathbb{R}^{\ell}}\|\mathbf{y}_{\ell}-\mathbf{R}_{\ell,\ell}\mathbf{d}\|_{2}^{2}+\mathbf{d}^{T}\mathbf{R}_{\ell,\ell}\mathbf{d}, \label{ModifiedPro}
\end{equation}
Obviously, the minimizer of \eqref{ModifiedPro} satisfies that $$(\mathbf{R}^{T}_{\ell,\ell}\mathbf{R}_{\ell,\ell}+\mathbf{R}_{\ell,\ell})\mathbf{d}-\mathbf{R}_{\ell,\ell}^{T}\mathbf{y}_{\ell}=\mathbf{0},$$
which has a solution
\begin{equation}
\mathbf{d}=(\mathbf{I}+\mathbf{R}_{\ell,\ell})^{-1}\mathbf{y}_{\ell}, \label{d_solution}
\end{equation}
where $\mathbf{I}$ is the identity matrix. Once $\mathbf{d}$ is solved, then the value of $f^{*}$ at $\mathcal{V}$ is $$\mathbf{f}^{*}=(f^{*}(\mathbf{v}_{1}),f^{*}(\mathbf{v}_{2}),\dots,f^{*}(\mathbf{v}_{n}))^{T}=\mathbf{R}_{\ell}\mathbf{d},$$ where $\mathbf{R}_{\ell}$ is the sub-matrix consisting of the first $\ell$ columns of $\mathbf{R}$.

At this point, we have obtained a solution of the original prediction model \eqref{Ori_model} by solving its equivalent model \eqref{Equivalent_model} based on a new RKHS $\mathcal{H}_{\mathcal{R}}$. Now, there are some important comments to make about two RKHSs $\mathcal{H}_{\mathcal{K}}$ and $\mathcal{H}_{\mathcal{R}}$.

\begin{itemize}
  \item First, as mentioned before, $\mathcal{H}_{\mathcal{K}}$ and $\mathcal{H}_{\mathcal{R}}$ have the same elements, but as RKHSs, they are equipped with different kernel functions $\mathcal{K}$ and $\mathcal{R}$ respectively.
  %\item Second, according to the representer theorem, the solution of the original prediction model $\eqref{Pro2}$ can be found in a subspace of $\mathcal{H}_{\mathcal{K}}$ (i.e., $\mathcal{H}_{\mathcal{K},\mathbf{X}}$). Similarly, the solution of the proposed model $\eqref{ModifiedPro1}$ can be found in a subspace of $\mathcal{H}_{\mathcal{R}}$ (see $\eqref{MyRepthm}$). It thus means that the proposed model allows us to find the solution of model $\eqref{Pro2}$ in a much low-dimensional subspace compared to $\mathcal{H}_{\mathcal{K},\mathbf{X}}$.
  \item Second, there is an important relationship between the two kernel gram matrixes $\mathbf{K}$ and $\mathbf{R}$, which is shown in the following theorem.
\end{itemize}

\begin{theorem}\label{KandR}
The kernel gram matrixes $\mathbf{K}$ and $\mathbf{R}$ have the following relationships:
\begin{equation}
\mathbf{R}=\mathbf{KT}, \label{KtoR}
\end{equation}
\end{theorem}
where $\mathbf{T}=(\lambda\mathbf{I}+\gamma\mathbf{LK})^{-1}$, where $\mathbf{I}$ is the identity matrix.
\begin{proof}
First, for $\forall f\in\mathcal{H}_{\mathcal{R}}$, we have $$\langle f,\mathcal{R}(\mathbf{v}_{i},\cdot)\rangle_{\mathcal{R}}=f(\mathbf{v}_{i}),~~\forall\mathbf{v}_{i}\in\mathcal{V}.$$
Let $\boldsymbol{\delta}_{i}$ denote the unit vector that all element are 0 except the $i$-th one, which is 1. Then,
   \begin{align*}
   &\langle f,\mathcal{R}(\mathbf{v}_{i},\cdot)-\sum_{s=1}^{n}\boldsymbol{\delta}_{s}^{T}\mathbf{T}\boldsymbol{\delta}_{i}\mathcal{K}(\mathbf{v}_{s},\cdot)\rangle_{\mathcal{R}}\\
    =&f(\mathbf{v}_{i})-\Big(\langle f,\sum_{s=1}^{n}\boldsymbol{\delta}_{s}^{T}\mathbf{T}\boldsymbol{\delta}_{i}\mathcal{K}(\mathbf{v}_{s},\cdot)\rangle_{\mathcal{R}}\Big)\\
    =&f(\mathbf{v}_{i})-\Big(\lambda \langle f,\sum_{s=1}^{n}\boldsymbol{\delta}_{s}^{T}\mathbf{T}\boldsymbol{\delta}_{i}\mathcal{K}(\mathbf{v}_{s},\cdot)\rangle_{\mathcal{K}}
    +\gamma\mathbf{f}^{T}\mathbf{LKT}\boldsymbol{\delta}_{i}\Big)\\
    =&f(\mathbf{v}_{i})-\Big(\lambda \sum_{s=1}^{n}\boldsymbol{\delta}_{s}^{T}\mathbf{T}\boldsymbol{\delta}_{i}f(\mathbf{v}_{s})+\gamma\mathbf{f}^{T}\mathbf{LKT}\boldsymbol{\delta}_{i}\Big)\\
    =&f(\mathbf{v}_{i})-\Big(\mathbf{f}^{T}\mathbf{T}\boldsymbol{\delta}_{i}+\gamma\mathbf{f}^{T}\mathbf{LKT}\boldsymbol{\delta}_{i}\Big)\\
    =&f(\mathbf{v}_{i})-\mathbf{f}^{T}(\lambda\mathbf{I}+\gamma\mathbf{LK})\mathbf{T}\boldsymbol{\delta}_{i}=f(\mathbf{v}_{i})-f(\mathbf{v}_{i})=\mathbf{0},\\
   \end{align*}
 which implies that
 $$\mathcal{R}(\mathbf{v}_{i},\cdot)=\sum_{s=1}^{n}\boldsymbol{\delta}_{s}^{T}\mathbf{T}\boldsymbol{\delta}_{i}\mathcal{K}(\mathbf{v}_{s},\cdot),~~\forall\mathbf{v}_{i}\in\mathcal{V}.$$
  Finally, since the kernel gram matrixes $\mathbf{R}$ and $\mathbf{K}$ are both symmetric,
 \begin{align*}
   \mathbf{R}(j,i)=\mathbf{R}(i,j)&=\mathcal{R}(\mathbf{v}_{i},\mathbf{v}_{j}) \\
   & =\sum_{s=1}^{n}\boldsymbol{\delta}_{s}^{T}\mathbf{T}\boldsymbol{\delta}_{i}\mathcal{K}(\mathbf{v}_{s},\mathbf{v}_{j}) \\
                                   &= \sum_{s=1}^{n}\mathbf{T}(s,i)\mathbf{K}(s,j)=\sum_{s=1}^{n}\mathbf{K}(j,s)\mathbf{T}(s,i)\\
                                   &=(\mathbf{KT})(j,i),
 \end{align*}
which implies that $\mathbf{R}=\mathbf{KT}$ and thus completes the proof.
 %In particular, let $f=\mathcal{R}(\mathbf{x}_{i},\cdot)-\sum_{j=1}^{n}\boldsymbol{\delta}_{j}^{T}\mathbf{T}\boldsymbol{\delta}_{i}\mathcal{K}(\mathbf{x}_{j},\cdot)$, then we have
% $$\|\mathcal{R}(\mathbf{x}_{i},\cdot)-\sum_{j=1}^{n}\boldsymbol{\delta}_{j}^{T}\mathbf{T}\boldsymbol{\delta}_{i}\mathcal{K}(\mathbf{x}_{j},\cdot)\|_{\mathcal{R}}=0.$$
% Thus, $$\mathcal{R}(\mathbf{x}_{i},\cdot)=\sum_{j=1}^{n}\boldsymbol{\delta}_{j}^{T}\mathbf{T}\boldsymbol{\delta}_{i}\mathcal{K}(\mathbf{x}_{j},\cdot),~~\forall\mathbf{x}_{i}\in\mathbf{X},$$
% which prove $\eqref{KtoR}$. Formula $\eqref{RtoK}$ can be proved by using a similar process.
 \end{proof}
 %Note that the symmetry of the kernel gram matrix is used in the proof of Theorem $\ref{KandR}$.
The above theorem intuitively states the relationship between $\mathbf{K}$ and $\mathbf{R}$ in which the matrix $\mathbf{T}=\lambda\mathbf{I}+\gamma\mathbf{LK}$ plays an important role. Note that $\mathbf{T}$ is always invertible since $\mathbf{L}$ and $\mathbf{K}$ are both positive semi-definite$\footnote{$\mathbf{LK}$ has non-negative eigenvalue since $\mathbf{L}$ and $\mathbf{K}$ are both positive semi-definite, which impies that $\mathbf{T}$ is invertible.}$.  Most importantly, formular \eqref{KtoR} provides the following two important results:
\begin{itemize}
  \item $\mathbf{R}$ can be computed by \eqref{KtoR} using $\mathbf{K}$ and $\mathbf{T}$.
  \item If $\mathbf{d}$ is solved by $\eqref{d_solution}$, $\mathbf{a}=\mathbf{T}_{\ell}\mathbf{d}$ is then a solution of the original prediction model $\eqref{Ori_model}$.

\end{itemize}
%an efficient way to compute $\mathbf{R}$ through $\mathbf{K}$ and $\mathbf{T}$. According to \eqref{KtoR},
 %Finally, we summarize the main comparisons of  $\mathcal{H}_{\mathcal{K}}$ and $\mathcal{H}_{\mathcal{R}}$ in Table $\ref{Com_two_RKHSs}$.

%\begin{table}[h]
%\renewcommand{\arraystretch}{1.5}
	%\centering
%\footnotesize
%\small
%	\begin{tabular}{c|c|c|c|c}
%    \hline
%    RKHSs %& inner product (RKHS norm) & inner pro. relationship
%    & \tabincell{c}{Kernel\\function} &\tabincell{c}{ Kernel function \\relationship} &\tabincell{c}{ Kernel gram \\matrix} & As a linear space\\
%    \hline
%   $\mathcal{H}_{\mathcal{K}}$ %& $\langle\cdot,\cdot\rangle_{\mathcal{K}}$ ($\|\cdot\|_{\mathcal{K}}$)& \multirow{2}{*}{$\eqref{New_inner_product}$}
%   &$\mathcal{K}$ & \multirow{2}{*}{Theorem $\ref{KandR}$} & \multirow{2}{*}{\tabincell{c}{$\mathbf{K}$=$\mathbf{R}\mathbf{T}^{-1}$ and \\$\mathbf{R}$=$\mathbf{KT}$}} &\multirow{2}{*}{$\mathcal{H}_{R}=\mathcal{H}_{\mathcal{K},\mathbf{X}}\subset\mathcal{H}_{\mathcal{K}}$} \\
 %  \cline{1-2}
%  $\mathcal{H}_{\mathcal{R}}$ %& $\langle\cdot,\cdot\rangle_{\mathcal{R}}$ ($\|\cdot\|_{\mathcal{R}}$) &
%  &$\mathcal{R}$ & & & \\
%   \hline
%	\end{tabular}
%\caption{Comparisons of two RKHSs $\mathcal{H}_{\mathcal{K}}$ and $\mathcal{H}_{\mathcal{R}}$}
% \label{Com_two_RKHSs}
%\end{table}

%Indeed, when substituting $\eqref{KtoR}$ into $\eqref{MyRepthm}$, we can get the solution of form $\eqref{Repthm}$ and verify that the coefficient vector $\mathbf{a}$ is now exactly $\mathbf{a}=\mathbf{T}_{\ell}\mathbf{d}$.

\subsection{Computational complexity}\label{complexity}
Based on Theorem \ref{KandR}, the detailed steps to solve the proposed equivalent model $\eqref{Equivalent_model}$ are shown in Algorithm $\ref{Alg_solution}$. Now, we provide a computational complexity analysis of Algorithm $\ref{Alg_solution}$.

First, step 1 requires $O(n^{3})$ computations to compute $\mathbf{R}=\mathbf{KT}$ by $\mathbf{K}$ and $\mathbf{T}=(\lambda\mathbf{I}+\gamma\mathbf{LK})^{-1}$ \cite{1996Matrix}. Second, solving $\mathbf{d}$ by $\eqref{d_solution}$ needs $O(\ell^{3})$ computations \cite{1996Matrix}. Finally,  $O(n\ell)$ computations are required to obtain $\mathbf{f}=\mathbf{R}_{\ell}\mathbf{d}$. To predict the complete score matrix $\mathbf{S}\in\mathbb{R}^{n\times m}$ with $n$ users and $m$ items, $\mathbf{R}$ only needs to be computed once since $\mathbf{L}$ and $\mathbf{K}$ are both fixed. Therefore, the proposed method requires $O(n^{3})$ computations to predict the first item, but it requires only $O(nm\ell+m\ell^{3})$ computations to predict the last $m-1$ items.
%the overall computational complexity to predict $m$ items is reduced to $O(n^{3}+(n\ell+\ell^{3})m)$.
%Therefore, while our method requires O(n3) computational complexity to predict the first item, it only requires O(m) time complexity to predict the remaining m-1 items

Now, there are some important comments to make about the three methods, namely: $\mathbf{Ori}$ given by solving $\eqref{a_solution}$, $\mathbf{GBa}$ proposed by \cite{2020Yang} and the proposed method given by Algorithm $\ref{Alg_solution}$.
\begin{itemize}
  \item First of all,  the three methods all provide a solution of the original prediction model $\eqref{Ori_model}$. However, $\mathbf{GBa}$ just provides an approximate solution, while both $\mathbf{Ori}$ and the proposed method find an optimal solution.
  \item Secondly, $\mathbf{Ori}$ searches the solution in a high-dimensional subspace when $n$ is large, which makes it require a high computational cost to predict all the items. To reduce the computational cost, $\mathbf{GBa}$ provides an approximate solution by solving $\eqref{a_solution}$ based on the bandlimited assumption, which makes it reduce the computational cost from $O(n^{3}m)$ to $O(n^{2}m(k_{b}+\ell))$ to predict the last $m-1$ items compared to $\mathbf{Ori}$.
  \item Finally, using the proposed equivalent prediction model, the proposed method searches the solution in a low-dimensional subspace ($\ell\ll n$). Since the kernel gram matrix $\mathbf{R}$ needs to be computed once, it further reduces the amount of computation from $O(n^{2}m(k_{b}+\ell))$ to $O(nm\ell+m\ell^{3})$ to predict the last $m-1$ items compared to $\mathbf{GBa}$. It thus means that the proposed method makes a qualitative improvement in terms of computational cost. Most importantly, unlike $\mathbf{GBa}$, the proposed method obtains an optimal solution of the original model \eqref{Ori_model}.
\end{itemize}

\begin{algorithm}[t]
\caption{Algorithm for solving the equivalent model}\label{Alg_solution}
\begin{algorithmic}[1]
\Require
$n$ users $\mathcal{V}=\{\mathbf{v}_{1},\mathbf{v}_{2},...,\mathbf{v}_{n}\}$ and $\ell$ labeled samples $X_{\ell}=\{(\mathbf{v}_{1},y_{1}),(\mathbf{v}_{2},y_{2}),...,(\mathbf{v}_{\ell},y_{\ell})\}$.
\Ensure
%Prediction function $f:\mathcal{V}\rightarrow\mathbb{R}$;
 \State Obtain $\mathbf{R}=\mathbf{KT}$ through $\mathbf{T}$;
\State Solve $\mathbf{d}=(d_{1},d_{2},...,d_{\ell})$ by $\eqref{d_solution}$;
\State \Return $f^{*}=\sum_{j=1}^{\ell}d_{j}\mathcal{R}(\mathbf{v}_{j},\cdot)$ and $\mathbf{f}^{*}=\mathbf{R}_{\ell}\mathbf{d}$;
\end{algorithmic}
\end{algorithm}

Table $\ref{Com_pre_model}$ shows the main comparisons of  $\mathbf{Ori}$, $\mathbf{GBa}$ and the proposed method.

 \begin{table}[htb]
\renewcommand{\arraystretch}{1.5}
	\centering
\footnotesize
%\small
	\begin{tabular}{c|c|c|c|c}
    \hline
    \multirow{2}{*}{Methods} & \multirow{2}{*}{Model}
    & \multirow{2}{*}{Solution form}  & \multicolumn{2}{c}{\multirow{1}{*}{Computational Complexity}}\\
    \cline{4-5}
    & &  & the first item & the last $m-1$ items\\
    \hline
   $\mathbf{Ori}$ & \multirow{2}{*}{model \eqref{Ori_model}} %& \multirow{2}{*}{$\mathcal{K}$} & \multirow{3}{*}{Theorem $\ref{KandR}$}
   & $\mathbf{f}^{*}=\mathbf{K}\mathbf{a}$  & $O(n^{3})$ & $O(n^{3}m)$\\
   \cline{1-1}\cline{3-5}
   $\mathbf{GBa}$ & & $\mathbf{f}^{*}\approx\mathbf{K}\mathbf{U}_{k_{b}}\mathbf{c}$ & $O(n^{3})$ & $O\Big(n^{2}m(k_{b}+\ell)\Big)$ \\
   \cline{1-5}
   %\cline{3-8}
   $\mathbf{Prop.}$ & model \eqref{Equivalent_model} & $\mathbf{f}^{*}=\mathbf{R}_{\ell}\mathbf{d}$  & $O(n^{3})$ & $O\Big(nm\ell+m\ell^{3}\Big)$\\
   \hline
	\end{tabular}
\caption{Comparison of the three methods in the prediction model with $n$ users and $m$ items.}
 \label{Com_pre_model}
\end{table}

\begin{figure*}[htb]
  \centering
    \subfigure[``two\_moons" point]{	
		\includegraphics[width=0.28\linewidth]{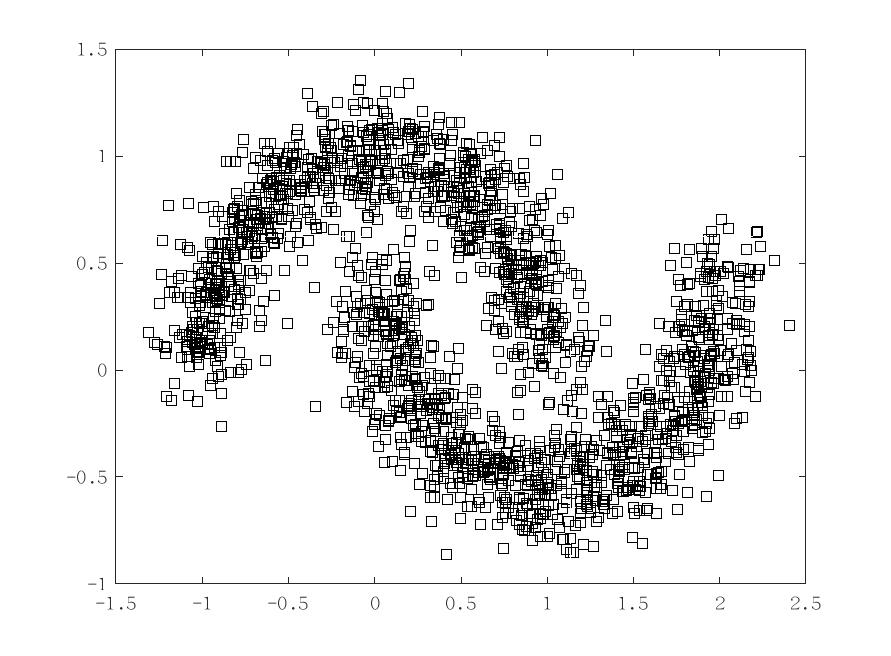}}\hspace{0.5cm}
    \subfigure[Labels of data on the graph]{	
		\includegraphics[width=0.28\linewidth]{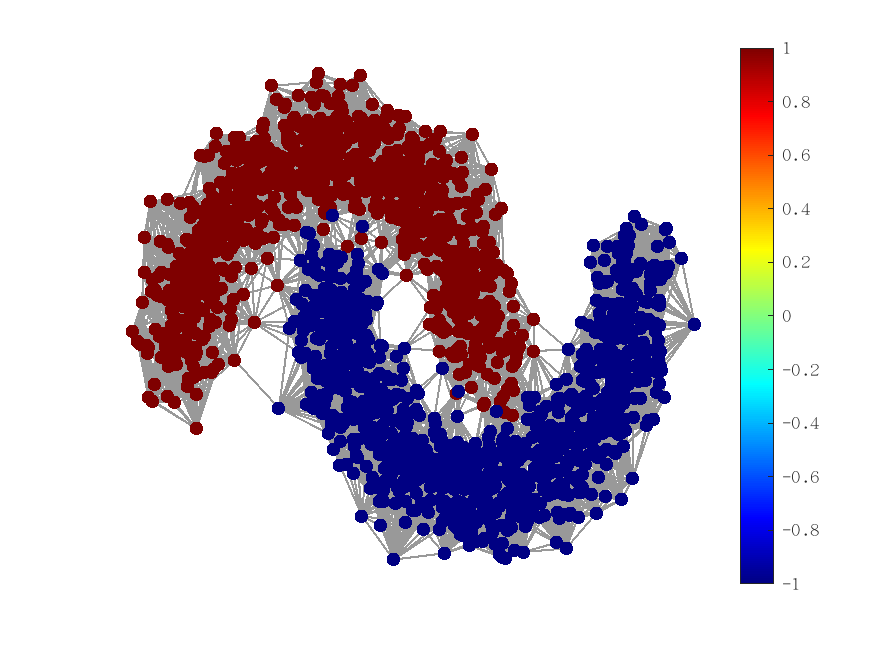}}\hspace{0.5cm}
    \subfigure[$\mathbf{Prop.}$]{	
		\includegraphics[width=0.28\linewidth]{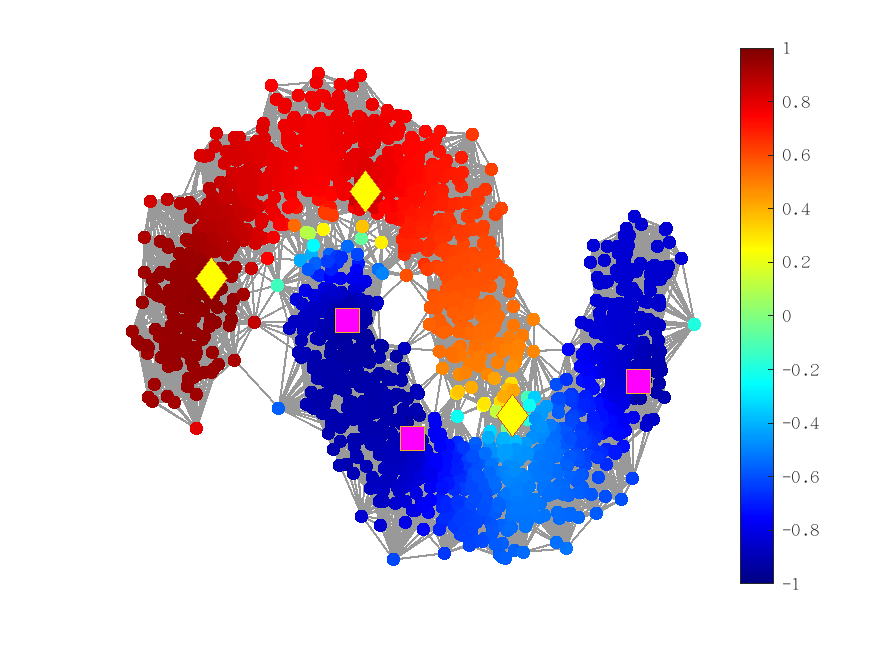}}

    \subfigure[$\mathbf{GBa10}$]{		
		\includegraphics[width=0.28\linewidth]{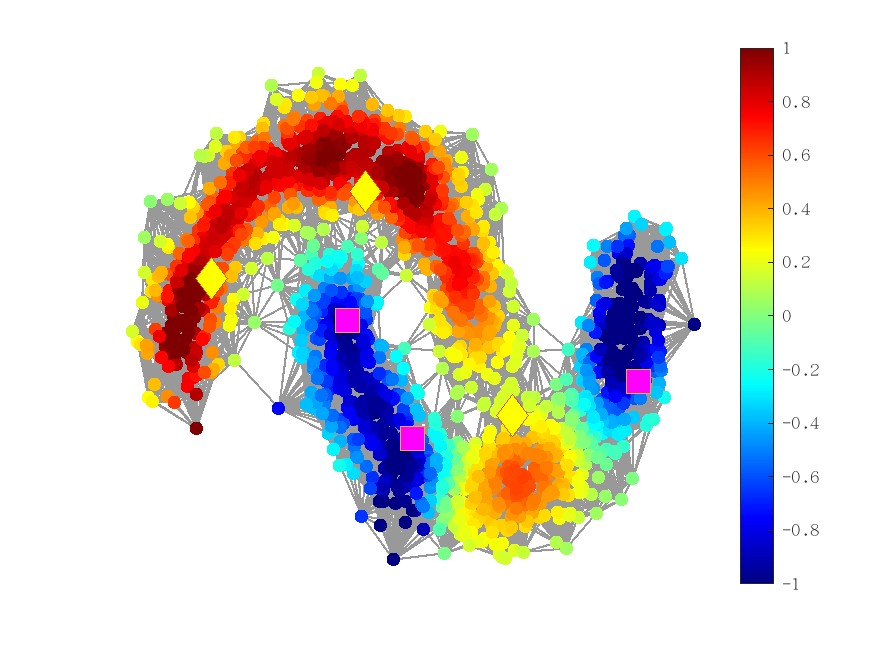}}\hspace{0.5cm}
    \subfigure[$\mathbf{GBa20}$]{		
		\includegraphics[width=0.28\linewidth]{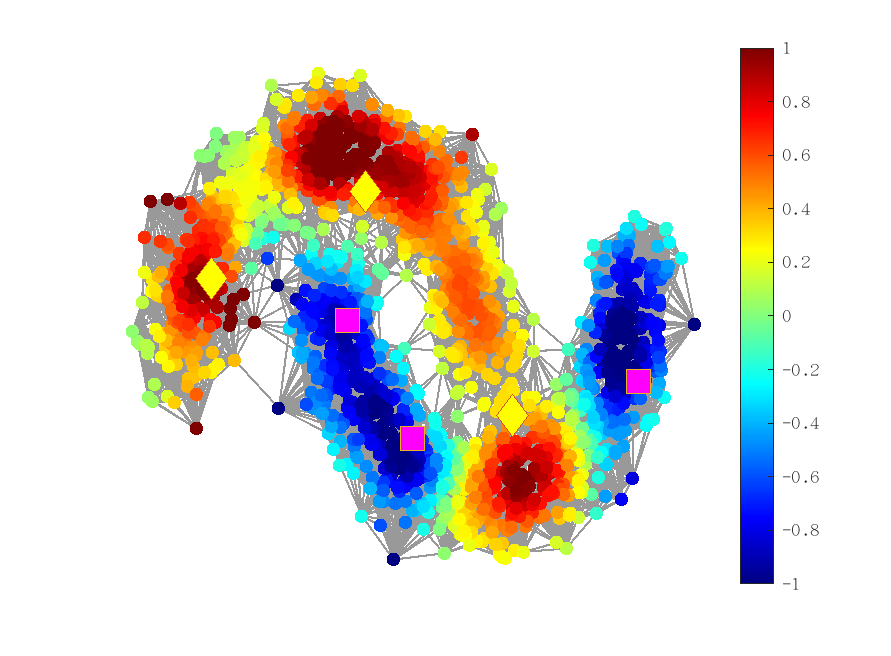}}\hspace{0.5cm}
    \subfigure[$\mathbf{GBa50}$]{		
		\includegraphics[width=0.28\linewidth]{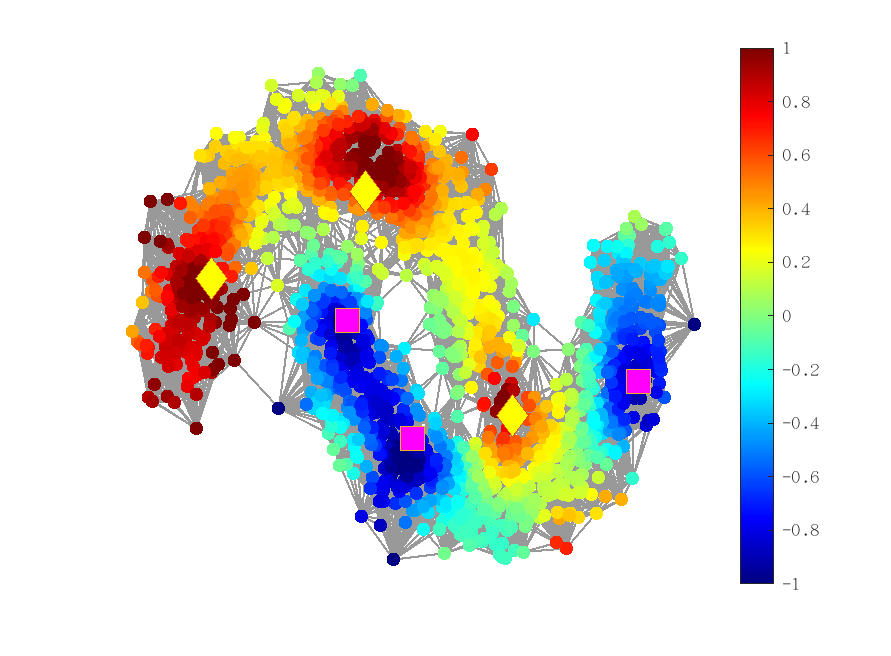}}

    \subfigure[$\mathbf{GBa100}$]{	
		\includegraphics[width=0.28\linewidth]{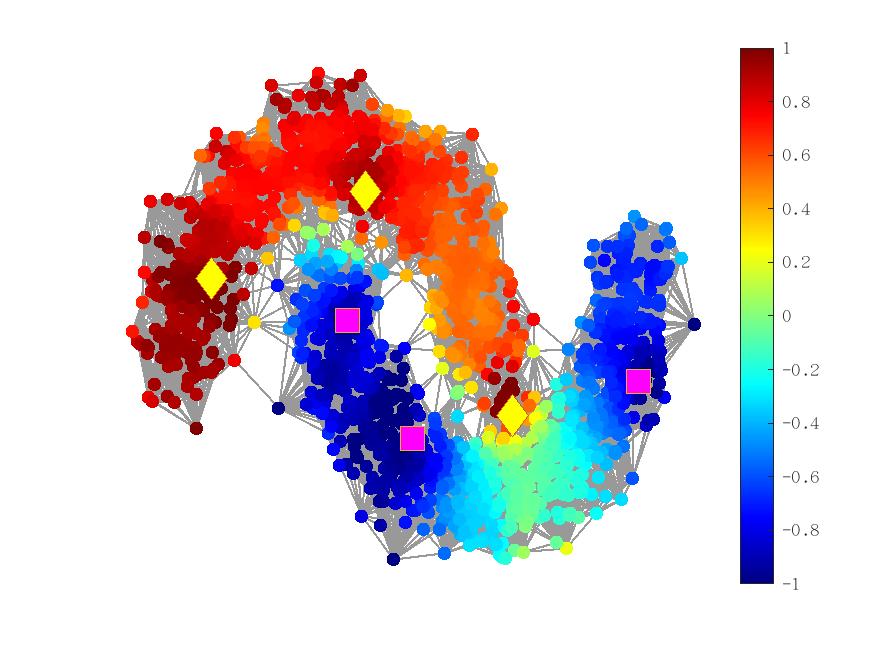}}\hspace{0.5cm}
    \subfigure[$\mathbf{MAE}$ vs $k_{b}$]{	
		\includegraphics[width=0.28\linewidth]{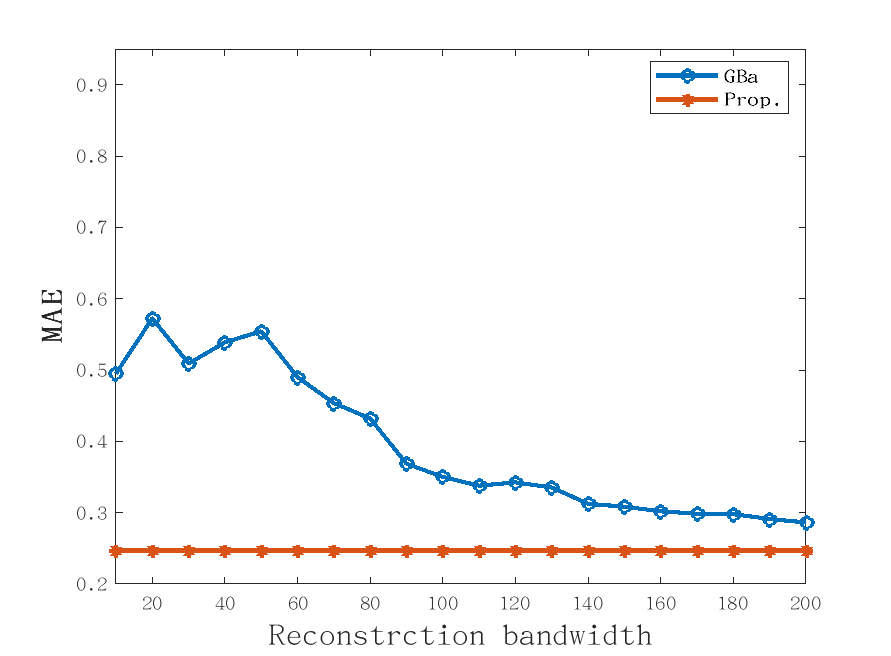}}

\caption{Two moons data: (a) ``two\_moons" data point on the 2D plane; (b) Labels of data on the graph: red means the label is 1 and blue means -1; (c)$\sim$(g) The prediction results of different methods in the case of $\ell=6$: yellow diamonds and rose-red squares represent data points with known labels of 1 and -1 respectively; and (h) The comparison of the prediction $\mathbf{MAE}$ of the four methods under different reconstruction bandwidth $k_{b}$ with $\ell=6$.} \label{Com_Manifold_learning}% and (h) the comparison of the prediction $\mathbf{MAE}$ of the four methods under different number of the known labels.} \label{Com_Manifold_learning}
\end{figure*}

\section{Final prediction algorithm}\label{Sec_fianl_prediction_alg}
In the previous section, we assumed that the kernel function $\mathcal{K}$ and the weighted matrix $\mathbf{W}$ are known. Therefore, we discuss their constructions in this section. Indeed, the construction of $\mathbf{W}$ and $\mathcal{K}$ should depend on the feature vectors $\mathcal{V}=\{\mathbf{v}_{1},\mathbf{v}_{2},\dots,\mathbf{v}_{n}\}$. In \cite{2020Yang}, the authors provided a good strategy to obtain the feature vectors $\mathcal{V}=\{\mathbf{v}_{1},\mathbf{v}_{2},\dots,\mathbf{v}_{n}\}$, so we adopt this strategy in this paper.

\subsection{Construction of kernel function $\mathcal{K}$}
Based on the feature vectors $\mathcal{V}=\{\mathbf{v}_{1},\mathbf{v}_{2},\dots,\mathbf{v}_{n}\}$, the kernel function $\mathcal{K}:\mathcal{V}\rightarrow\mathbb{R}$ with the kernel gram matrix $\mathbf{K}=\Big(\mathcal{K}(\mathbf{v}_{i},\mathbf{v}_{j})\Big)_{\mathbf{v}_{i},\mathbf{v}_{j}\in\mathcal{V}}$ is often chosen as the Gaussian kernel
\begin{equation}\label{Gaussian_K}
\mathcal{K}(\mathbf{v}_{i},\mathbf{v}_{j})=e^{-\|\mathbf{v}_{i}-\mathbf{v}_{j}\|_{2}^{2}/2\sigma^{2}},~~\mathbf{v}_{i},\mathbf{v}_{j}\in \mathcal{V},
\end{equation}
with a constant parameter $\sigma>0$. In \cite{2020Yang}, the linear kernel function was proved to perform poorly in the prediction, and thus we do not consider such a kernel function in this paper.

\subsection{Construction of $\mathbf{W}$ on the graph}
The next issue is to construct the graph's adjacent matrix $\mathbf{W}$. In this paper, we adopt the strategy to construct $\mathbf{W}$ based on the assumption that two users have a large affinity if they are close in the feature space. The effectiveness of this strategy can be referred in the literature \cite{2006Manifold}.

A common choice of $\mathbf{W}$ is the heat kernel weights
\begin{equation}\label{Cons_W}
W_{ij}=e^{-\|\mathbf{v}_{i}-\mathbf{v}_{j}\|_{2}^{2}/4\epsilon},~~\mathbf{v}_{i},\mathbf{v}_{j}\in\mathcal{V},
\end{equation}
with a constant parameter $\epsilon>0$. Sometime we also perform a $k$-nearest neighbor sparse process to obtain a sparse graph. Once $\mathbf{W}$ is constructed, $\mathbf{L}=\mathbf{D}-\mathbf{W}$ can then be obtained immediately. Indeed, such a $k$-nearest graph can be obtained by the function ``\emph{gsp\_nn\_graph}'' in the GSP toolbox \cite{2016GSPBOX}.

Note that the strategy of constructing $\mathbf{W}$ we adopt in this paper is quite different from the two commonly strategies existing in \cite{2015GraphLearning} and \cite{2020Yang}. We let the edge weights of graph $\mathcal{G}$ be determined by the distance between two users in the feature space. Furthermore, once the feature vectors $\mathcal{V}=\{\mathbf{v}_{1},\mathbf{v}_{2},\dots,\mathbf{v}_{n}\}$ are obtained, $\mathbf{W}$ can be easily constructed by $\eqref{Cons_W}$, which requires much less time than the other two methods in \cite{2015GraphLearning} and \cite{2020Yang}. Most importantly, experimental results in Section $\ref{Sec_exp}$ show that $\mathbf{W}$ constructed in this way performs well in the terms of prediction accuracy.

\subsection{Final algorithm for prediction}
Finally, the final algorithm to predict a score matrix $\mathbf{S}$ for $n$ users and $m$ items can be organized as Algorithm \ref{Final_pre}.

Note that the kernel gram matrix $\mathbf{R}$ in Algorithm $\ref{Alg_solution}$ only need to be computed once, while $\mathbf{d}$ will be computed repeatedly $m$ times based on different known labels to predict the all $m$ items.

\begin{algorithm}[t]
\caption{Final algorithm to predict a score matrix $\mathbf{S}$ for $n$ users and $m$ items}\label{Final_pre}
\begin{algorithmic}[1]
\Require
Given a sparse rating matrix $\mathbf{S}$ and two regularization parameters $\lambda$ and $\gamma$.
\Ensure
%Prediction function $f:\mathcal{V}\rightarrow\mathbb{R}$;
 \State Generate the feature vectors $\mathcal{V}=\{\mathbf{v}_{1},\mathbf{v}_{2},\dots,\mathbf{v}_{n}\}$ based on \cite{2020Yang};
\State Choose $\mathcal{K}$ as a Gaussian kernel by $\eqref{Gaussian_K}$ with the corresponding kernel gram matrix $\mathbf{K}=\Big(\mathcal{K}(\mathbf{v}_{i},\mathbf{v}_{j})\Big)_{\mathbf{v}_{i},\mathbf{v}_{j}\in\mathcal{V}}$;
\State  Obtain the adjacent matrix $\mathbf{W}$ of the $k$-nearest sparse graph constructed by using the function ``\emph{gsp\_nn\_graph}'' in the GSP toolbox \cite{2016GSPBOX} based on the the feature vectors $\mathcal{V}$, and compute the graph Laplacian $\mathbf{L}$.
%\State Construct the graph's adjacent matrix $\mathbf{W}$ by $\eqref{Cons_W}$ and perform the sparse process to obtain a $k$-nearest neighbor sparse graph with a sparse $\mathbf{W}$. Then, compute the graph Laplacian $\mathbf{L}$;
\State Let $X_{\ell}=\{(\mathbf{v}_{1},y_{1}),(\mathbf{v}_{2},y_{2}),...,(\mathbf{v}_{\ell},y_{\ell})\}$ be the set of labeled samples for a given item, then we solve the prediction model by using Algorithm $\ref{Alg_solution}$ and obtain $\mathbf{f}^{*}=\mathbf{R}_{\ell}\mathbf{d}$. This process will be repeated $m$ times to predict the whole $\mathbf{S}$.
\State\Return $\mathbf{S}$.
\end{algorithmic}
\end{algorithm}

\section{Numerical experiments}\label{Sec_exp}
In this section, we conduct several experiments to test the performance of the proposed method (labeled $\mathbf{Prop}.$). Experiments are performed on the following three data sets:
\begin{itemize}
  \item Two moons data: ``two\_moons" point cloud data with 2000 data points which is commonly used in the GSP toolbox \cite{2016GSPBOX}.
%  \item MovieLens-100k data: MovieLens-100k data set consists of 100000 ratings (1$\sim$5) from 943 users on 1682 movies. Each user has rated at least 20 movies. This dataset generates two groups data sets \cite{2020Yang}:
    \item MovieLens-100k data: MovieLens-100k data set contains 100000 ratings (1$\sim$5) given by 943 users for 1682 movies. At least 20 movies are rated by each user. This dataset generates two groups data sets \cite{2020Yang}:

     % $1)$ u1.base$\sim$u5.base and u1.test$\sim$u5.test. Each pair of u$\ast$.base and u$\ast$.test is a $80\%\setminus20\%$ splits of the whole dataset. There are disjoint test sets between any two pairs.
    $1)$ u1.base$\sim$u5.base and u1.test$\sim$u5.test. Each pair of u$\ast$.base-u$\ast$.test is a $80\%\setminus20\%$ splits of the whole dataset. Any two pairs have disjoint test sets.

      %$2)$ ua.base, ua.test and ub.base, ub.test. Each pair of sets is generated by splitting the whole data set into a training set and a test set with exactly 10 ratings per user in the test set. The sets ua.test and ub.test are also disjoint.
       $2)$ ua.base, ua.test and ub.base, ub.test. Each pair is obtained by splitting the whole data set into a training set and a test set.  The two test set (i.e., ua.test and ub.test) contains exactly 10 ratings per user, and they are also disjoint.

  %\item Netflix data: Netflix data set contains over 100 million ratings (1$\sim$5) from 480189 randomly-chosen, anonymous Netflix customers over 17770 thousand movie titles. The percentage of known ratings is about 1.12\% \cite{2020Yang}.
  \item Netflix data: Netflix data set consists of over 100 million ratings (1$\sim$5). These ratings are given by 480189 Netflix customers on 17770 thousand movie titles. In this dataset, about 1.12\% of the ratings are known \cite{2020Yang}.

\end{itemize}
 For the $\mathbf{GBa}$ method, we set different reconstruction bandwidth $k_{b}$ (i.e., $k_{b}=10,20,50,100$) to test its performance relative to $k_{b}$. We label the $\mathbf{GBa}$ methods corresponding to $k_{b}=10,20,50,100$ as $\mathbf{GBa10}$, $\mathbf{GBa20}$,$\mathbf{GBa50}$, $\mathbf{GBa100}$ respectively. As a performance metric, the mean absolute error ($\mathbf{MAE}$) will be used to measure the prediction error in this paper. When comparing the computational efficiency of different methods, we only compute the running time of predicting the whole score matrix $\mathbf{S}$ with known $\mathbf{K}$ and $\mathbf{L}$. Furthermore, we should point out that all the programs run on a Dell T7920 workstation (Intel(R) Xeon(R) Gold 5122 CPU @ 1.70 GHz, 62 GB RAM).

\subsection{Synthetic Data: Two moons data}
In this section, we test the performance of different methods on the prediction problem of the two moons data$\footnote{Note that we remove the $\mathbf{Ori}$ method from the comparison since it has the same solution as the proposed method.}$. Figure \ref{Com_Manifold_learning}(a) shows the distribution of the two moons data on the 2D plane$\footnote{At this point, the feature vector of each data point corresponds to its coordinate on the 2D plane.}$. In the experiments, we first use the function ``\emph{gsp\_nn\_graph}'' \cite{2016GSPBOX} to obtain the $k$-nearest ($k=30$) sparse graph of the data points. And we artificially treat the first 1000 data points as one class (labeled as 1) and the last 1000 data points as another class (labeled as -1), as shown in Figure \ref{Com_Manifold_learning}(b). For prediction, we empirically set the parameters to be $\sigma=0.1$, $\lambda=0.0001$ and $\gamma=0.005$.

 First, we fix the number of known labels ($\ell=6$) to conduct the experiments. Specifically, we mark the labels of three data points in each class as known, i.e., yellow diamonds and rose-red squares shown in Figure \ref{Com_Manifold_learning}(c)$\sim$(g), and then predict the labels of other data points based on these six known labels. The prediction results are shown in Figure \ref{Com_Manifold_learning}(c)$\sim$(g). Combined with the real labels of data shown in Figure \ref{Com_Manifold_learning}(b), it can be observed from Figure \ref{Com_Manifold_learning}(c)$\sim$(g) that the proposed method predicts the labels of all data points more accurately than all $\mathbf{GBa}$ methods.

 To show the performance of different methods more intuitively, we compute the prediction $\mathbf{MAE}$ for all methods, and then show the results as a function of $k_{b}$ in Figure \ref{Com_Manifold_learning}(h). We observe that the curves of the proposed method is horizontal since it is independent of $k_{b}$. Furthermore, it can be seen that as $k_{b}$ increases, the prediction $\mathbf{MAE}$ of $\mathbf{GBa}$  gets smaller and closer to that of the proposed method. This is because the approximate solution obtained by $\mathbf{GBa}$ gets closer and closer to the optimal solution of the original model as $k_{b}$ increases. Finally, it is worth emphasizing that the proposed method always perform the better than $\mathbf{GBa}$.

\begin{figure*}[!t]
  \centering
    \subfigure[]{	
		\includegraphics[width=0.23\linewidth]{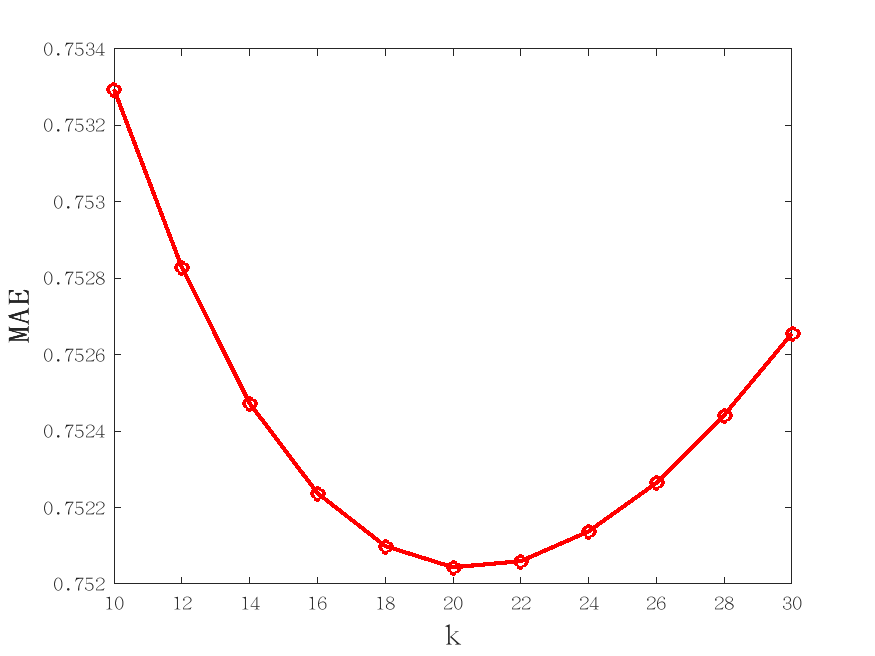}}
    \subfigure[]{	
		\includegraphics[width=0.23\linewidth]{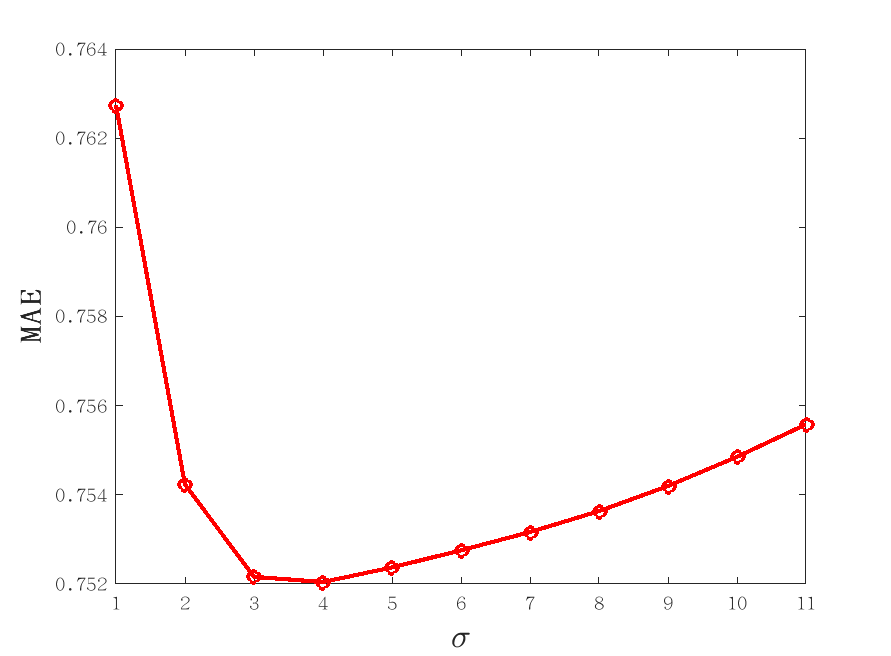}}
    \subfigure[]{	
		\includegraphics[width=0.23\linewidth]{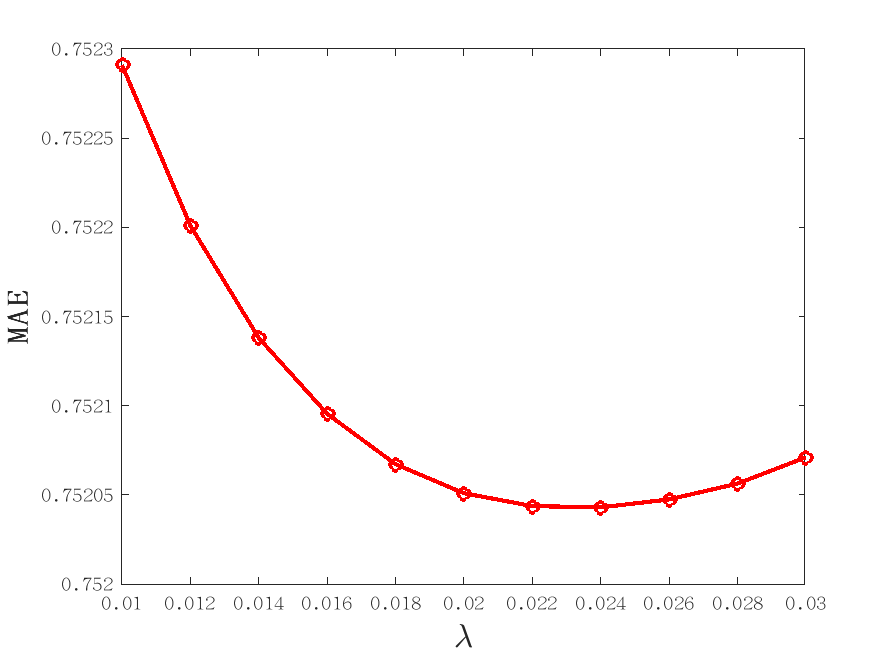}}%\hspace{2cm}
    \subfigure[]{		
		\includegraphics[width=0.23\linewidth]{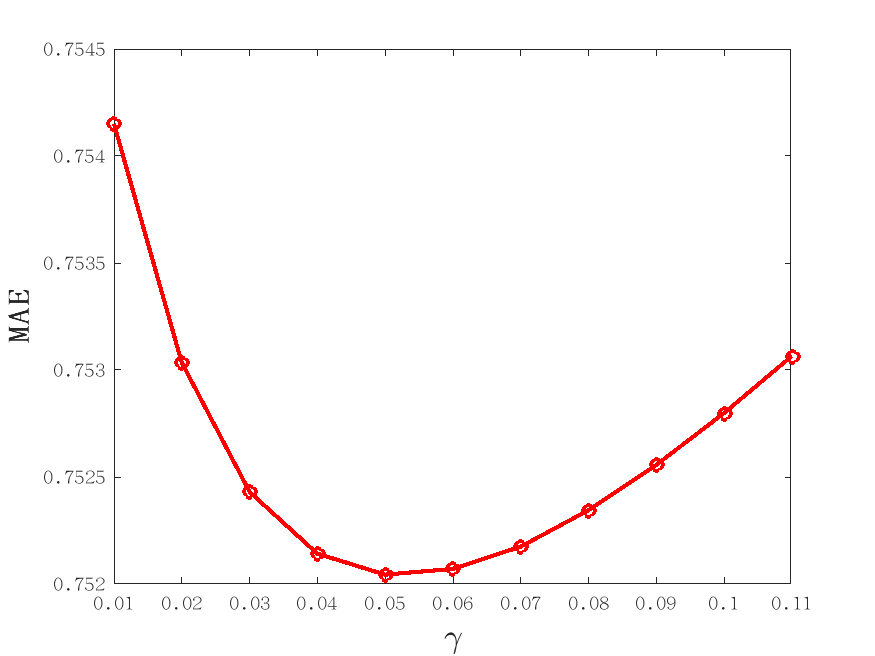}}

\caption{Learning of user-based prediction model parameters $k,\sigma,\lambda,\gamma$: (a) $\sigma=4$, $\lambda=0.022$, $\gamma=0.05$, (b) $k=20$, $\lambda=0.022$, $\gamma=0.05$, (c) $k=20$, $\sigma=4$, $\gamma=0.05$, and (d) $k=20$, $\sigma=4$, $\lambda=0.022$.} \label{Parameter_learning_user}
\end{figure*}

\begin{figure*}[!t]
  \centering
    \subfigure[User-based]{	
		\includegraphics[width=0.25\linewidth]{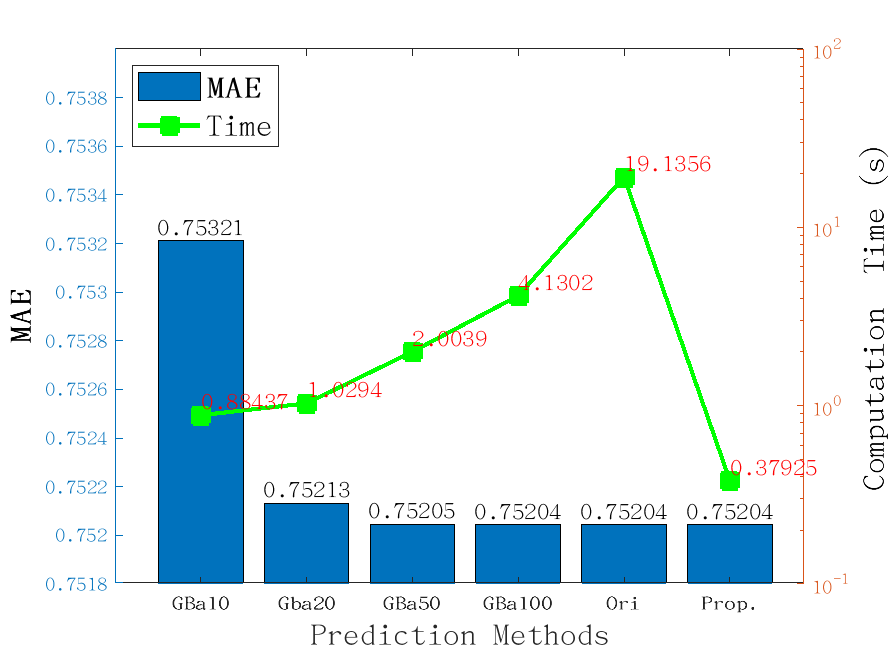}}\hspace{2cm}
    \subfigure[Item-based]{
        \includegraphics[width=0.25\linewidth]{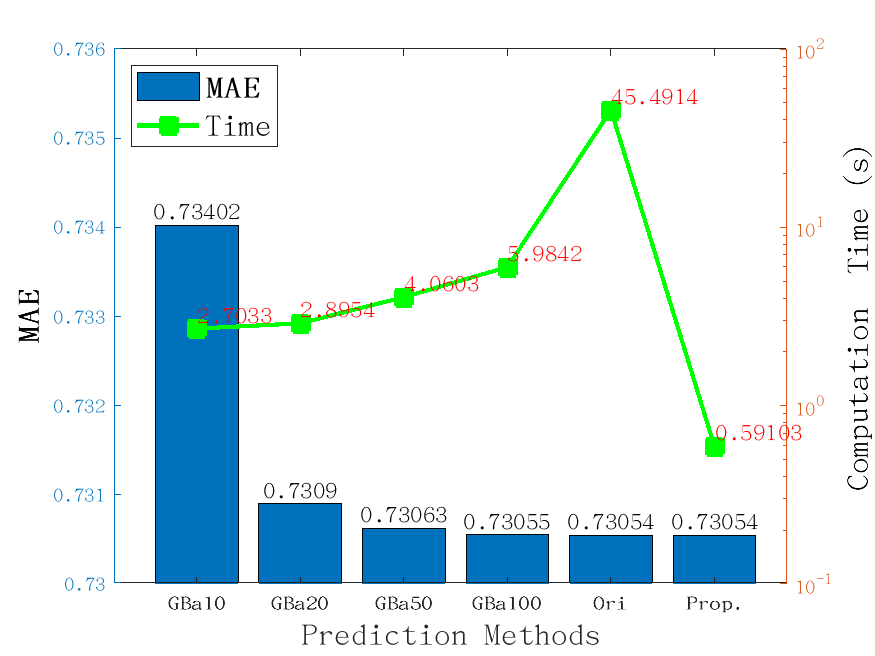}}
\caption{ Prediction results on the u1.$\ast\sim$u5.$\ast$ data sets: the histogram corresponds to $\mathbf{MAE}$ (i.e., the left coordinate axis), and the polyline corresponds to the running time (i.e., the right coordinate axis).} \label{Prediction_on_u1_u5}
\end{figure*}

\begin{figure*}[!t]
  \centering
    \subfigure[User-based]{	
		\includegraphics[width=0.25\linewidth]{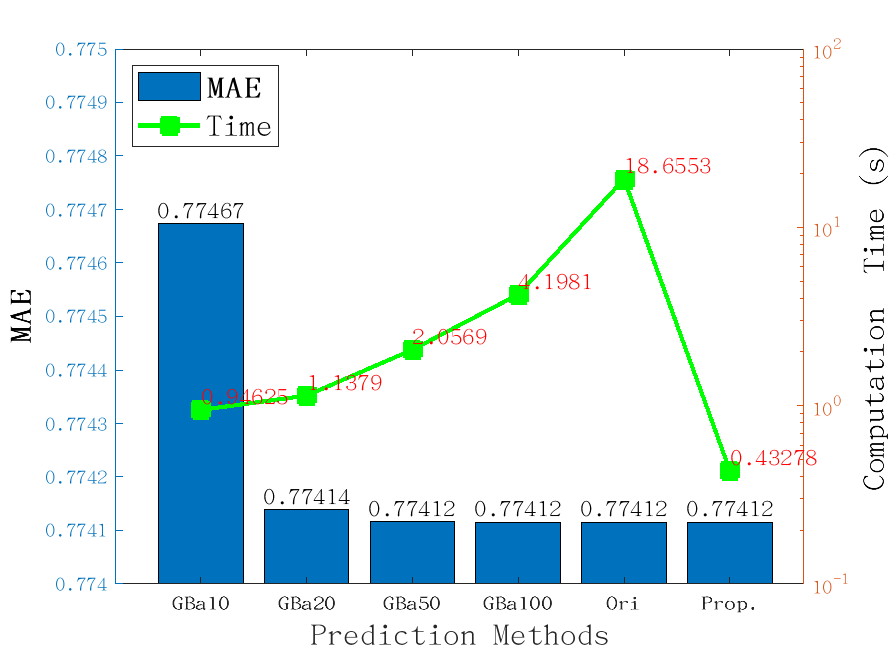}}\hspace{2cm}
    \subfigure[Item-based]{
        \includegraphics[width=0.25\linewidth]{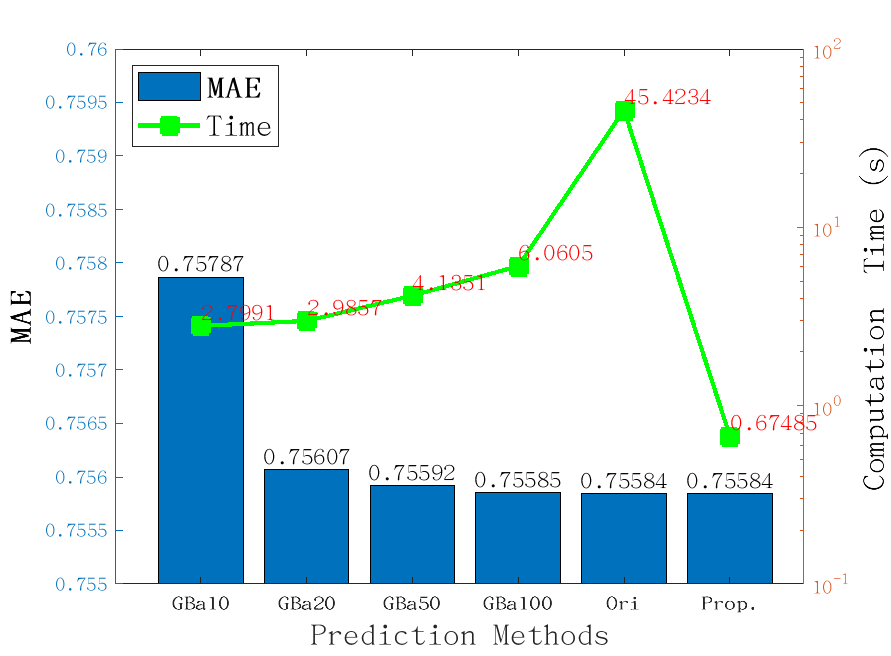}}
\caption{ Prediction results on the ua.$\ast\sim$ub.$\ast$ data sets: the histogram corresponds to $\mathbf{MAE}$ (i.e., the left coordinate axis), and the polyline corresponds to the running time (i.e., the right coordinate axis).} \label{Prediction_on_ua_ub}
\end{figure*}

\subsection{Prediction on MovieLens-100k data}\label{Sec_pre_MovieLens100K}
In this section, we run the experiments on the MovieLens-100k data set. %The experiments are conducted into two groups.
Two groups of experiments are conducted as follows.
The first group trains on the data sets u1.base$\sim$u5.base and tests on u1.test$\sim$u5.test. The second group run on the data sets ua.base, ua.test and ub.base, ub.test.

%Since the proposed method choose a different strategy from that of \cite{2020Yang} to construct the adjacency matrix $\mathbf{W}$,
Since the construction of $\mathbf{W}$ in the proposed method is different from that in \cite{2020Yang}, it needs to relearn the parameters in the prediction model. Therefore, we first use the data sets u1.$\ast\sim$u5.$\ast$ to select the best parameters for prediction.
%For each pair of subsets, we construct the feature vectors $\mathcal{V}$ using u$\ast$.base. Then, we predict the entries of u$\ast$.test and compute the absolute error between the estimated ratings and the real ones in u$\ast$.test.
For each pair of subsets, we use u$\ast$.base to construct the feature vectors $\mathcal{V}$.
%Then, we predict the entries of u$\ast$.test and compute the absolute error between the estimated ratings and the real ones in u$\ast$.test.
We then predict the entries in u$\ast$.test and compute the absolute error between the predicted ratings and the ground true ones.
We compute the global ($\mathbf{MAE}$) across the five test sets to select the best parameters$\footnote{Note that when constructing the adjacency matrix $\mathbf{W}$ using the function ``\emph{gsp\_nn\_graph}'', we only learn the best $k$ to obtain the $k$-nearest sparse graph,  while other parameters are set as the default parameters of ``\emph{gsp\_nn\_graph}''.}$. The results of parameter learning of user-based prediction are shown in Figure \ref{Parameter_learning_user}. Finally, we select the parameters as $k=20$, $\sigma=4$, $\lambda=0.022$, $\gamma=0.05$ for both user-based and item-based predictions.

Using the above learned parameters, we conduct experiments on u1.$\ast\sim$u5.$\ast$ to compare the proposed method with other methods. In the experiments, we compute the global $\mathbf{MAE}$ between the predicted ratings and the ground true ones in u$\ast$.test. The results of the global $\mathbf{MAE}$ and the mean running time are shown in Figure $\ref{Prediction_on_u1_u5}$. From the histogram in Figure $\ref{Prediction_on_u1_u5}$, we observe that the proposed method and $\mathbf{Ori}$ have the same global $\mathbf{MAE}$ and perform the best in terms of prediction accuracy, where the global $\mathbf{MAE}$ reaches 0.75204 and 0.73054 for user-based and item-based predictions respectively. It can also be obseved that as $k_{b}$ increases, the global $\mathbf{MAE}$ of $\mathbf{GBa}$ gets smaller and closer to that of $\mathbf{Ori}$ and the proposed method. The reason for this is that the approximate solution of $\mathbf{GBa}$ is closer to the optimal solution of the original model as $k_{b}$ increases. However, at the same time, the polyline in Figure $\ref{Prediction_on_u1_u5}$ shows that $\mathbf{GBa}$ requires more computation time as $k_{b}$ increases. It thus means that compared to $\mathbf{GBa}$, the proposed method significantly reduces the computational cost, while accurately solving the original model and maintaining a good prediction accuracy.

%Furthermore, we observe that the proposed method and $\mathbf{SVM}$ are the fastest among all method. However, $\mathbf{SVM}$ performs more poorly than the proposed method in terms of prediction accuracy, as can be seen from the corresponding histogram in Figure $\ref{Prediction_on_u1_u5}$.

\begin{table*}[!t]
%\renewcommand{\arraystretch}{1.5}
%\caption{\footnotesize COMPARISON OF COMPLEXITY OF DIFFERENT SAMPLING SET SELECTION ALGORITHMS FOR SELECTING $D$ NODES.}

	\centering
\footnotesize
%\small
	\begin{tabular}{c|c|cccccc|}
    \hline
    \multicolumn{2}{c|}{\multirow{2}{*}{$\mathbf{MAE}$}} & \multicolumn{6}{c|}{User-first}\\
    \cline{3-8}
        \multicolumn{2}{c|}{~}
         & $\mathbf{GBa10}$ & $\mathbf{GBa20}$ & $\mathbf{GBa50}$ & $\mathbf{GBa100}$ & $\mathbf{Ori}$ & $\mathbf{Prop}.$ \\
   %      & $\mathbf{GBa10}$ & $\mathbf{GBa20}$ & $\mathbf{GBa50}$ & $\mathbf{GBa100}$ & $\mathbf{Ori}$ & $\mathbf{Prop}.$\\
    \hline

    \multirow{5}{*}{\textbf{Ub}} & 1\%
          & 0.7793  &  0.7787 & $\mathbf{0.7785}$ & $\mathbf{0.7785}$ & $\mathbf{0.7785}$  & $\mathbf{0.7785}$  \\
     %     & 0.7720  &  0.7715 & $\mathbf{0.7714}$ & $\mathbf{0.7714}$ & $\mathbf{0.7714}$  & $\mathbf{0.7714}$ \\

          & 2\%
          & 0.7484  &  0.7469 & $\mathbf{0.7467}$ & $\mathbf{0.7467}$ & $\mathbf{0.7467}$  & $\mathbf{0.7467}$  \\
      %    & 0.7456  &  0.7445 & $\mathbf{0.7443}$ & $\mathbf{0.7443}$ & $\mathbf{0.7443}$  & $\mathbf{0.7443}$ \\

          & 3\%
          & 0.7352  &  0.7339 & $\mathbf{0.7337}$ & $\mathbf{0.7337}$ & $\mathbf{0.7337}$  & $\mathbf{0.7337}$  \\
     %     & 0.7348  &  0.7332 & $\mathbf{0.7330}$ & $\mathbf{0.7330}$ & $\mathbf{0.7330}$  & $\mathbf{0.7330}$ \\

          & 4\%
          & 0.7265  &  0.7247 & $\mathbf{0.7244}$ & $\mathbf{0.7244}$ & $\mathbf{0.7244}$  & $\mathbf{0.7244}$  \\
    %      & 0.7271  &  0.7249 & $\mathbf{0.7246}$ & $\mathbf{0.7246}$ & $\mathbf{0.7246}$  & $\mathbf{0.7246}$ \\

          & 5\%
          & 0.7240  &  0.7214 & $\mathbf{0.7211}$ & $\mathbf{0.7211}$ & $\mathbf{0.7211}$  & $\mathbf{0.7211}$   \\
     %     & 0.7247  &  0.7228 & $\mathbf{0.7225}$ & $\mathbf{0.7225}$ & $\mathbf{0.7225}$  & $\mathbf{0.7225}$ \\
      \hline

    \multirow{5}{*}{\textbf{Ib}} & 1\%
          & 0.7894  &  0.7892 & 0.7893 & $\mathbf{0.7892}$ & $\mathbf{0.7892}$  & $\mathbf{0.7892}$   \\
     %     & 0.7554  &  $\mathbf{0.7559}$ & $\mathbf{0.7559}$ & $\mathbf{0.7559}$ & $\mathbf{0.7559}$  & $\mathbf{0.7559}$ \\

          & 2\%
          & 0.7483  &  0.7467 & 0.7464 & $\mathbf{0.7463}$ & $\mathbf{0.7463}$  & $\mathbf{0.7463}$   \\
     %     & 0.7292  &  0.7275 & 0.7272 & $\mathbf{0.7271}$ & $\mathbf{0.7271}$  & $\mathbf{0.7271}$ \\

          & 3\%
          & 0.7316  &  0.7287 & 0.7282 & $\mathbf{0.7281}$ & $\mathbf{0.7281}$  & $\mathbf{0.7281}$   \\
     %     & 0.7188  &  0.7154 & 0.7151 & $\mathbf{0.7149}$ & $\mathbf{0.7149}$  & $\mathbf{0.7149}$ \\

          & 4\%
          & 0.7192  &  0.7165 & 0.7160 & $\mathbf{0.7159}$ & $\mathbf{0.7159}$  & $\mathbf{0.7159}$   \\
     %     & 0.7101  &  0.7077 & 0.7071 & 0.7070 & $\mathbf{0.7069}$  & $\mathbf{0.7069}$ \\

          & 5\%
          & 0.7120  &  0.7113 & 0.7108 & $\mathbf{0.7107}$ & $\mathbf{0.7107}$  & $\mathbf{0.7107}$    \\
     %     & 0.7056  &  0.7036 & 0.7030 & $\mathbf{0.7028}$ & $\mathbf{0.7028}$  & $\mathbf{0.7028}$ \\

    \hline
    \multicolumn{2}{c|}{\multirow{2}{*}{$\mathbf{MAE}$}} & \multicolumn{6}{c|}{Item-first}\\
    \cline{3-8}
    \multicolumn{2}{c|}{~}
    & $\mathbf{GBa10}$ & $\mathbf{GBa20}$ & $\mathbf{GBa50}$ & $\mathbf{GBa100}$ & $\mathbf{Ori}$ & $\mathbf{Prop}.$\\
    \hline
    \multirow{5}{*}{\textbf{Ub}} & 1\%
     %     & 0.7793  &  0.7787 & $\mathbf{0.7785}$ & $\mathbf{0.7785}$ & $\mathbf{0.7785}$  & $\mathbf{0.7785}$
          & 0.7720  &  0.7715 & $\mathbf{0.7714}$ & $\mathbf{0.7714}$ & $\mathbf{0.7714}$  & $\mathbf{0.7714}$ \\

          & 2\%
    %      & 0.7484  &  0.7469 & $\mathbf{0.7467}$ & $\mathbf{0.7467}$ & $\mathbf{0.7467}$  & $\mathbf{0.7467}$
          & 0.7456  &  0.7445 & $\mathbf{0.7443}$ & $\mathbf{0.7443}$ & $\mathbf{0.7443}$  & $\mathbf{0.7443}$ \\

          & 3\%
     %     & 0.7352  &  0.7339 & $\mathbf{0.7337}$ & $\mathbf{0.7337}$ & $\mathbf{0.7337}$  & $\mathbf{0.7337}$
          & 0.7348  &  0.7332 & $\mathbf{0.7330}$ & $\mathbf{0.7330}$ & $\mathbf{0.7330}$  & $\mathbf{0.7330}$ \\

          & 4\%
     %     & 0.7265  &  0.7247 & $\mathbf{0.7244}$ & $\mathbf{0.7244}$ & $\mathbf{0.7244}$  & $\mathbf{0.7244}$
          & 0.7271  &  0.7249 & $\mathbf{0.7246}$ & $\mathbf{0.7246}$ & $\mathbf{0.7246}$  & $\mathbf{0.7246}$ \\

          & 5\%
    %      & 0.7240  &  0.7214 & $\mathbf{0.7211}$ & $\mathbf{0.7211}$ & $\mathbf{0.7211}$  & $\mathbf{0.7211}$
          & 0.7247  &  0.7228 & $\mathbf{0.7225}$ & $\mathbf{0.7225}$ & $\mathbf{0.7225}$  & $\mathbf{0.7225}$ \\
      \hline

    \multirow{5}{*}{\textbf{Ib}} & 1\%
     %     & 0.7894  &  0.7892 & 0.7893 & $\mathbf{0.7892}$ & $\mathbf{0.7892}$  & $\mathbf{0.7892}$
          & 0.7554  &  $\mathbf{0.7559}$ & $\mathbf{0.7559}$ & $\mathbf{0.7559}$ & $\mathbf{0.7559}$  & $\mathbf{0.7559}$ \\

          & 2\%
     %     & 0.7483  &  0.7467 & 0.7464 & $\mathbf{0.7463}$ & $\mathbf{0.7463}$  & $\mathbf{0.7463}$
          & 0.7292  &  0.7275 & 0.7272 & $\mathbf{0.7271}$ & $\mathbf{0.7271}$  & $\mathbf{0.7271}$ \\

          & 3\%
    %      & 0.7316  &  0.7287 & 0.7282 & $\mathbf{0.7281}$ & $\mathbf{0.7281}$  & $\mathbf{0.7281}$
          & 0.7188  &  0.7154 & 0.7151 & $\mathbf{0.7149}$ & $\mathbf{0.7149}$  & $\mathbf{0.7149}$ \\

          & 4\%
    %      & 0.7192  &  0.7165 & 0.7160 & $\mathbf{0.7159}$ & $\mathbf{0.7159}$  & $\mathbf{0.7159}$
          & 0.7101  &  0.7077 & 0.7071 & 0.7070 & $\mathbf{0.7069}$  & $\mathbf{0.7069}$ \\

          & 5\%
     %     & 0.7120  &  0.7113 & 0.7108 & $\mathbf{0.7107}$ & $\mathbf{0.7107}$  & $\mathbf{0.7107}$
          & 0.7056  &  0.7036 & 0.7030 & $\mathbf{0.7028}$ & $\mathbf{0.7028}$  & $\mathbf{0.7028}$ \\

    \hline
	\end{tabular}
\caption{Prediction results on Netflix data set: the known labels accounts for about 1\% of each data subset. \textbf{Ub} and \textbf{Ib} represent user-based and item-based prediction respectively.}
 \label{Pre_MAE_Netflix}
\end{table*}

\begin{figure*}[!t]
  \centering
    \subfigure[User-first: user-based]{	
		\includegraphics[width=0.23\linewidth]{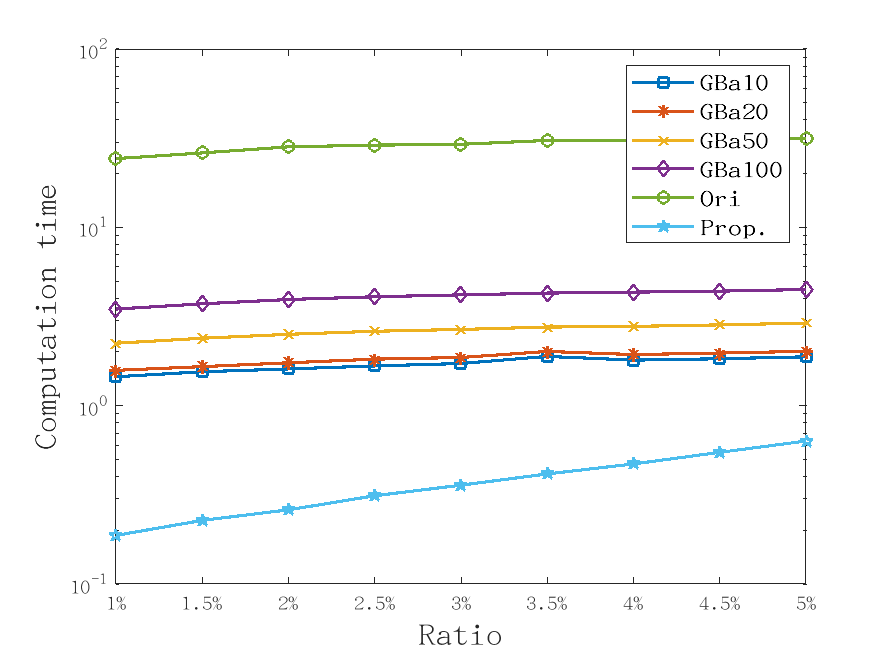}}
    \subfigure[User-first: item-based]{
        \includegraphics[width=0.23\linewidth]{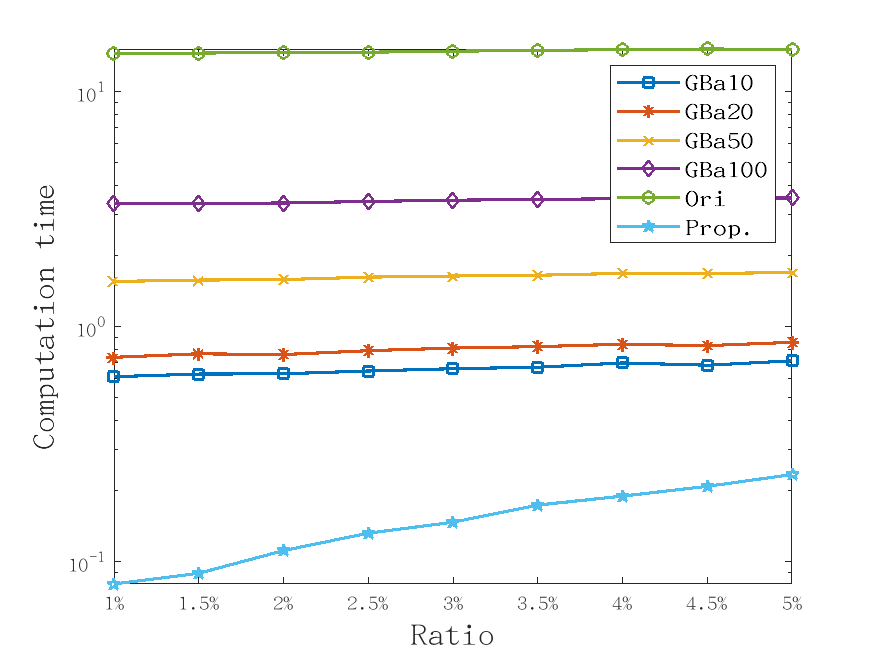}}
    \subfigure[Item-first: user-based]{	
		\includegraphics[width=0.23\linewidth]{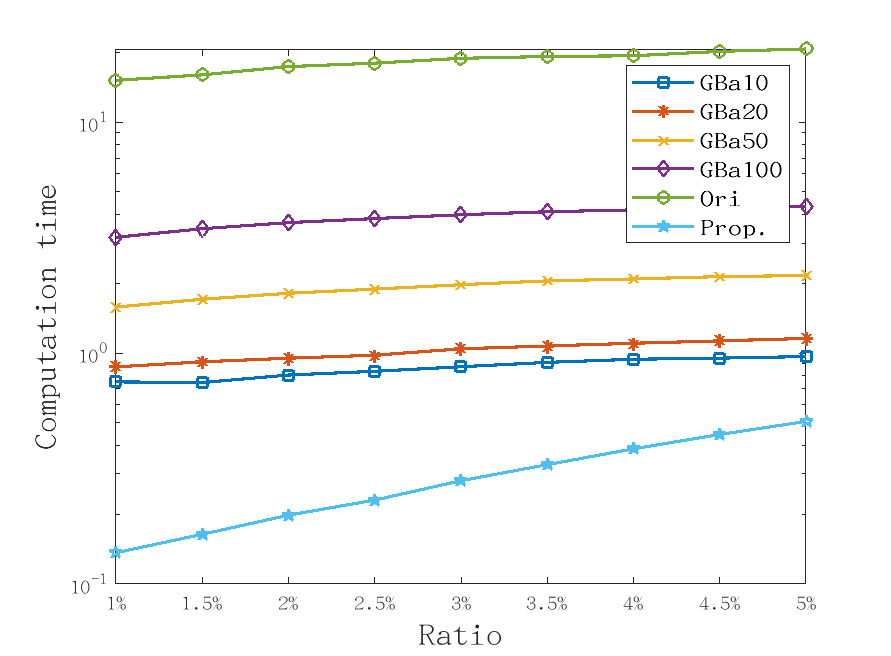}}
    \subfigure[Item-first: item-based]{
        \includegraphics[width=0.23\linewidth]{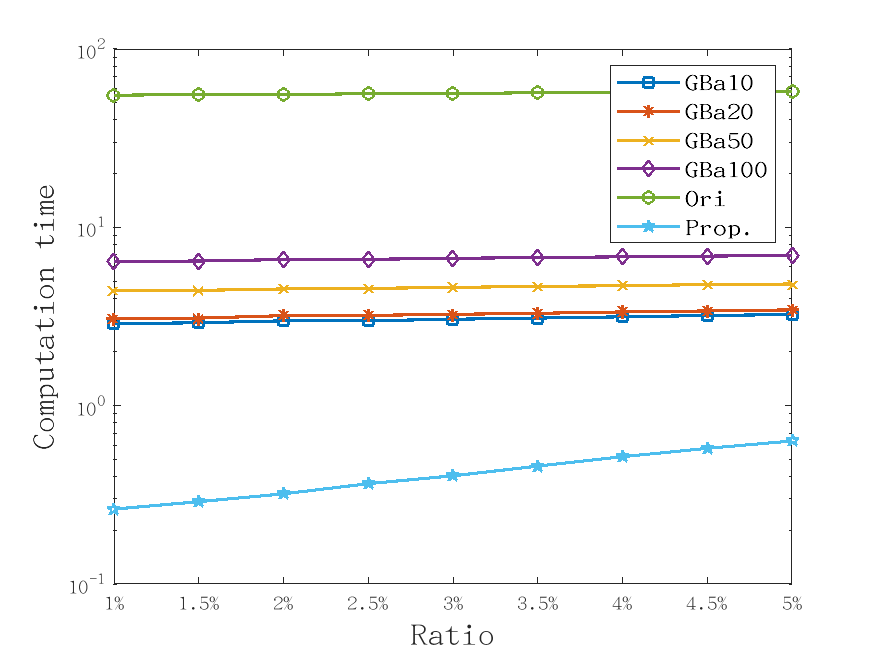}}
\caption{ Comparisons of the running time under different ratios of known labels on the Netflix data sets.} \label{Pre_time_Netflix}
\end{figure*}

Finally, to investigate the generalization ability of the proposed method, we further use the same learned parameters to conduct several experiments on ua.$\ast$ and ub.$\ast$ data sets. The global $\mathbf{MAE}$ and the mean running time are shown in Figure $\ref{Prediction_on_ua_ub}$. It can be easily observed that the proposed method runs faster that other methods and performs best in terms of prediction accuracy, reaching 0.77412 for user-based prediction and 0.75584 for item-based prediction. In a word, the results on ua.$\ast$ and ub.$\ast$  are similar to those on u1.$\ast\sim$u5.$\ast$, thus highlighting the robustness of the proposed method to different rating matrices.

\begin{table*}[!t]
 \centering
\begin{tabular}{c|c|ccccc}
    \hline

     \multicolumn{2}{c|}{\multirow{2}{*}{Speedup factor}} & \multirow{2}{*}{$\mathbf{GBa10}$} & \multirow{2}{*}{$\mathbf{GBa20}$}
     & \multirow{2}{*}{$\mathbf{GBa50}$}& \multirow{2}{*}{$\mathbf{GBa100}$} & \multirow{2}{*}{$\mathbf{Ori}$}\\
      \multicolumn{2}{c|}{~}&~&~&~& &~\\
    \hline

   \multicolumn{1}{c|}{\multirow{2}{*}{User-first}} & user-based  &  16.73   &  17.31  &  22.11 & 29.89 & 640.59   \\
    \cline{2-7}
                                & item-based  & 21.96  &  23.03 & 29.77 & 39.81 & 870.79   \\
      \hline
    \multirow{2}{*}{Item-first} & user-based  & 18.62  & 19.36 & 25.15  & 33.95 & 754.44   \\
    \cline{2-7}
                                & item-based  & 21.14 & 21.55  & 27.08 & 35.73  & 876.39 \\
    \hline
	\end{tabular}
\caption{Speedup factors of our method with respect to the alternative methods.}
\label{Speedup_factor}
\end{table*}

\subsection{Prediction on Netflix}
Next, we further conduct several experiments on Netflix data to estimate the robustness of the proposed method against different data sets.

The experiments are conducted using the same data used in \cite{2020Yang} which is divided into two groups. The first group is the item-first data set, which contains five data subsets and each data subset contains 1777 items and 1000 users.  The other group is the user-first data set, which also contains five data subsets while each data subset contains 888 items and 1500 users. In the experiments, the known ratings accounting for about 1\% of each data subset are randomly selected to form the training set, and the rest serves as test set.

Using the same parameters as Section \ref{Sec_pre_MovieLens100K}, we conduct the experiments on the user-first and item-first data sets respectively. The results of the global $\mathbf{MAE}$ and the mean running time of prediction each data subset are shown in Table \ref{Pre_MAE_Netflix} and Figure \ref{Pre_time_Netflix} respectively. From Table \ref{Pre_MAE_Netflix}, it can be easily seen that all methods have better prediction $\mathbf{MAE}$ as the ratio of the known labels increases, and the proposed method always outperforms other methods. Furthermore, Figure \ref{Pre_time_Netflix} shows that in terms of computational cost, the proposed method are significantly faster than other methods. Specially, Table \ref{Pre_MAE_Netflix} and Figure \ref{Pre_time_Netflix} together confirm again that as $k_{b}$ increases, $\mathbf{GBa}$ performs better but requires more computation time.

\begin{table*}[!t]
\renewcommand{\arraystretch}{1.5}
	\centering

	\begin{tabular}{c|ccccccc}
    \hline
    \multirow{2}{*}{Methods} & \multirow{2}{*}{$\textrm{UB}_{\textrm{BCF}}$} & \multirow{2}{*}{$\textrm{UB}_{\textrm{HUS}}$} & \multirow{2}{*}{SVD} & \multirow{2}{*}{SVD++}
    & \multirow{2}{*}{PMF} &\multirow{2}{*}{BPMF} & \multirow{2}{*}{FPMF} \\
    %\cline{9-10}\cline{11-12}
       &   &    &   &   &   &   &   \\
    \hline
    Ave. $\mathbf{MAE}$ & 0.7751 & 0.7870 & 0.8147 & 0.7864 & 0.7370 & 0.7614 & 0.7318 \\
    Std.  & 0.0042 & 0.0019 & 0.0115 & 0.0069 & 0.0082 & 0.0047 & 0.0088 \\
   \hline
   \multirow{2}{*}{Methods}  & \multicolumn{2}{c}{$\mathbf{GBa}$ in \cite{2020Yang}}  & \multicolumn{2}{c}{$\mathbf{Prop}.$}\\
    \cline{2-3}\cline{4-5}
        & \textbf{Ub} &\textbf{Ib}  & \textbf{Ub}&\textbf{Ib} \\
    \hline
    Ave. $\mathbf{MAE}$  & 0.7546 & 0.7314  & 0.7520 &$\mathbf{0.7305}$\\
    Std.  & 0.0056 & 0.0054  & 0.0053 &0.0055\\
   \hline
	\end{tabular}
\caption{Comparison results of prediction $\mathbf{MAE}$ with popular methods on the u1.$\ast\sim$u5.$\ast$ data sets. \textbf{Ub} and \textbf{Ib} represent user-based and item-based prediction respectively.}
 \label{Com_pop_Methods}
\end{table*}

\subsection{Speedup factor for running time on large-scale data}
Finally, in order to further compare the running time of each method on large-scale data (i.e., both $n$ and $m$ are both large), we merge the above five item-first datasets to obtain one item-first dataset. It contains 4,597 users and 8,885 items with 629,003 known ratings accounting for about 1.54\% of the total. At the same time, the five user-first datasets are merged into one user-first dataset, which contains 6,620 users and 4,440 items with 462,106 known ratings accounting for about 1.57\% of the total. Then, we test the running time of each method on the two merged datasets separately. In the experiments, the known ratings accounting for about 1\% of each dataset are randomly selected as the known labels to predict other ratings.

To compare the running time more conveniently, we take the proposed method as a benchmark, and compute the Speedup factor with respect to the alternative methods, i.e,
$$\textrm{Speedup~factor}:=\frac{\textrm{Running~time~of~the~alternative~method}}{\textrm{Running~time~of~the~proposed~method}}.$$

Finally, Table \ref{Speedup_factor} shows the comparison results of running time on the two merged data sets. From the results of Table \ref{Speedup_factor} we observe that, among all methods, the proposed method runs faster than other methods. Specifically, the proposed method is $>15$ times faster than $\mathbf{GBa10}$,  $>16$ times faster than $\mathbf{GBa20}$, $>20$ times faster than $\mathbf{GBa50}$, $>25$ times faster than $\mathbf{GBa100}$ and $>600$ times faster than $\mathbf{Ori}$. It thus confirms again that the proposed method outperforms $\mathbf{GBa}$ and $\mathbf{Ori}$ much more as $k_{b}$ increases. %Furthermore, combined with the previous experimental results, the proposed method greatly reduces the computational cost of solving the original prediction model as $n$ and $m$ increase, while maintaining a good prediction accuracy.

\subsection{Experiment analysis}
In this section, we left with some important comments in terms of prediction accuracy and computational cost on the experiment results.

$\mathbf{Prediction~accuracy}$: First, the experiment results show that, the proposed model \eqref{Equivalent_model} solves the original prediction model accurately, and thus obtains the same prediction accuracy as $\mathbf{Ori}$. We also observe that as $k_{b}$ increases, the approximate solution of $\mathbf{GBa}$ is closer to the solution of the original model, but at the same time it requires more computation time.

Furthermore, we compare the proposed method with other popular prediction methods on the u1.$\ast\sim$u5.$\ast$ data sets in terms of prediction accuracy, namely: $\textrm{UB}_{\textrm{BCF}}$ \cite{2015Patra}, $\textrm{UB}_{\textrm{HUS}}$ \cite{2017Wang}, SVD \cite{1990Indexing}, SVD++ \cite{2008Yehuda}, PMF \cite{2007PMF}, BPMF \cite{2008BPMF}, FPMF \cite{2020Feng} and $\mathbf{GBa}$ in \cite{2020Yang}. We directly quoted the experimental results of other methods from \cite{2020Yang}. Finally, the comparison results of prediction $\mathbf{MAE}$ are shown in Table \ref{Com_pop_Methods}. It can be easily observed that the prediction $\mathbf{MAE}$ of the proposed method can arrives at 0.7305 for item-based prediction which is superior to other methods.

$\mathbf{Computational~cost}$: The most worth mentioning is that the proposed method significantly reduces the computational cost of solving the original model \eqref{Ori_model} compared to $\mathbf{Ori}$ and $\mathbf{GBa}$ methods. As stated in Section $\ref{Sec_modifiedmodel}$, to predict ratings with $n$ users and $m$ items, the proposed method reduces the computational cost from $O(n^{2}m(k_b+\ell))$ ($O(n^{3}m)$) to $O(nm\ell+m\ell^3)$ to predict the last $m-1$ items compared to $\mathbf{GBa}$ ($\mathbf{Ori}$). Specifically, as shown in Table \ref{Speedup_factor}, the proposed method can drop the running time down $>15$ times compared to $\mathbf{GBa}$ and $>600$ times compared to $\mathbf{Ori}$. Furthermore, we can observe that using the proposed method, the reduction in computational cost is more significant for larger $n$ and $m$.

%Compared to $\mathbf{Ori}$ (in seconds):
%47.23$\rightarrow$0.98, 112.53$\rightarrow$2.95, 42.84$\rightarrow$1.16, 109.05$\rightarrow$3.32, 73.30$\rightarrow$2.10, 34.91$\rightarrow$0.69, 46.70$\rightarrow$1.53, 135.43$\rightarrow$3.92
%Compared to $\mathbf{GBa}$ (in seconds):
%25.09$\rightarrow$0.98, 66.03$\rightarrow$2.95, 25.21$\rightarrow$1.16, 66.23$\rightarrow$3.32, 44.23$\rightarrow$2.10, 18.29$\rightarrow$0.69, 26.53$\rightarrow$1.53, 81.87$\rightarrow$3.92

\section{Conclusion}\label{Sec_conclusion}
This paper proposes an equivalent prediction model to solve the original prediction model based on graph Laplacian regularization in recommendation system (RS). The proposed model allows us to find a solution of the original prediction model in a much low-dimensional subspace. Based on the proposed equivalent prediction model, an efficient method is designed skillfully to solve the original model accurately, and thus reduces the computational cost from $O(n^2m(k_b+\ell))$ to $O(nm\ell+m\ell^3)$ to predict the last $m-1$ items compared to the graph-based approximate method in \cite{2020Yang}. Finally, we propose a final algorithm based on the proposed equivalent prediction model, and the experimental results on the synthetic data and two commonly used real-world data sets show that the proposed method also maintains a good prediction accuracy.

\section*{Declaration of Competing Interest}
The authors declare that they have no known competing financial interests or personal relationships that could have
appeared to influence the work reported in this paper.
\section*{Acknowledgement}
This research was partially supported by National Natural Science Foundation of China (Nos: 12171488),
Guangdong Province Key Laboratory of Computational Science at the Sun Yat-sen University (2020B1212060032), and the Research Grants Council of the Hong Kong Special Administrative Region, China, under Project C1013-21GF.

%and thus significantly reduces the computational cost of predicting the last $m-1$ items from
%$O(n^2m(k_b+\ell))$ to $O(nm\ell+m\ell^3)$ compared to the graph-based approximate method in \cite{2020Yang}. The superiority of the proposed method in term of cost efficiency makes it more practical in the real-world applications.

%Furthermore, based on the proposed equivalent prediction model, an efficient method is designed skillfully to solve the original optimization accurately. Finally, the experimental results on synthetic dataset and two commonly used real-world datasets show that our method not only has advantages in term of cost efficiency, but also has good prediction accuracy.

\section*{References}
  \bibliographystyle{elsarticle-num}
  \biboptions{square,numbers,sort&compress}
 % \bibliography{mybibfile}
 \bibliography{elsarticle-template-num-names}

\end{document}